\documentclass[11pt,letterpaper]{article}

\usepackage[utf8]{inputenc}
\usepackage[T1]{fontenc}
\usepackage{lmodern}
\usepackage[DIV=11]{typearea} 

\usepackage{microtype}

\usepackage{amssymb}
\usepackage{amsmath}
\usepackage{amsthm}
\usepackage{thmtools}
\usepackage{thm-restate}

\usepackage[procnumbered,ruled,vlined,linesnumbered]{algorithm2e}

\usepackage{xcolor}
\usepackage{xspace}
\usepackage{xfrac}

\usepackage{comment}

\usepackage[backend=biber, style=alphabetic, backref=true, doi=false, url=false, maxcitenames=3, mincitenames=3, maxbibnames=10, minbibnames=10, sortlocale=en_US]{biblatex}

\usepackage[ocgcolorlinks]{hyperref} 
\usepackage{cleveref}

\DeclareUnicodeCharacter{221A}{$\sqrt{}$}

\colorlet{DarkRed}{red!50!black}
\colorlet{DarkGreen}{green!50!black}
\colorlet{DarkBlue}{blue!50!black}

\hypersetup{
	linkcolor = DarkRed,
	citecolor = DarkGreen,
	urlcolor = DarkBlue,
	bookmarksnumbered = true,
	linktocpage = true
}

\declaretheorem[numberwithin=section]{theorem}
\declaretheorem[numberlike=theorem]{lemma}

\declaretheorem[numberlike=theorem]{corollary}
\declaretheorem[numberlike=theorem]{definition}
\declaretheorem[numberlike=theorem]{claim}

\DontPrintSemicolon
\SetKw{KwAnd}{and}
\SetProcNameSty{textsc}
\SetFuncSty{textsc}

\title{Decremental Strongly-Connected Components and Single-Source Reachability in Near-Linear Time}
\author{
Aaron Bernstein\thanks{Rutgers University New Brunswick, Department of Computer Science}
\and
Maximilian Probst\thanks{BARC, University of Copenhagen, Universitetsparken 1, Copenhagen 2100, Denmark, The author is supported by Basic Algorithms Research Copenhagen (BARC), supported by Thorup's Investigator Grant from the Villum Foundation under Grant No. 16582.}
\and
Christian Wulff-Nilsen\thanks{Department of Computer Science, University of Copenhagen. This research is supported by the Starting Grant 7027-00050B from the Independent Research Fund Denmark under the Sapere Aude research career programme.}
}

\date{\today}

\hypersetup{
	pdftitle = {Linear time algorithm for maintaining SCCs},
	pdfauthor = {Aaron Bernstein, Maximilian Probst, Christian Wulff-Nilsen}
}

\begin{document}
\maketitle
\begin{abstract}
Computing the Strongly-Connected Components (SCCs) in a graph $G=(V,E)$ is known to take only $O(m + n)$ time using an algorithm by Tarjan from 1972[SICOMP 72] where $m = |E|$, $n=|V|$. For fully-dynamic graphs, conditional lower bounds provide evidence that the update time cannot be improved by polynomial factors over recomputing the SCCs from scratch after every update. Nevertheless, substantial progress has been made to find algorithms with fast update time for \emph{decremental} graphs, i.e. graphs that undergo edge deletions.

In this paper, we present the first algorithm for general decremental graphs that maintains the SCCs in total update time $\tilde{O}(m)$\footnote{We use $\tilde{O}(f(n))$ notation to suppress logarithmic factors, i.e. $g(n) = \tilde{O}(f(n))$ if $g(n) = O(f(n) \text{polylog}(n))$.}, thus only a polylogarithmic factor from the optimal running time. Previously such a result was only known for the special case of planar graphs [Italiano et al, STOC 17]. Our result should be compared to the formerly best algorithm for general graphs achieving $\tilde{O}(m\sqrt{n})$ total update time by Chechik et al. [FOCS 16] which improved upon a breakthrough result leading to $O(mn^{0.9 + o(1)})$ total update time by Henzinger, Krinninger and Nanongkai [STOC 14, ICALP 15]; these results in turn improved upon the longstanding bound of $O(mn)$ by Roditty and Zwick [STOC 04]. 

All of the above results also apply to the decremental Single-Source Reachability (SSR) problem which can be reduced to decrementally maintaining SCCs. A bound of $O(mn)$ total update time for decremental SSR was established already in 1981 by Even and Shiloach [JACM 81].

Using a well known reduction, we use our decremental result to achieve new update/query-time trade-offs in the fully dynamic setting. We can maintain the  reachability of pairs $S \times V$, $S \subseteq V$ in fully-dynamic graphs with update time $\tilde{O}(\frac{|S|m}{t})$ and query time $O(t)$ for all $t \in [1,|S|]$; this matches the best All-Pairs Reachability algorithm for $S = V$ [Łącki, TALG 13].
\end{abstract}
\newpage

\section{Introduction}

For a directed graph $G=(V,E)$, with $n = |V|, m=|E|$, the Strongly-Connected Components (SCCs) of $G$ are the sets of the unique partition of the vertex set $V$ into sets $V_1, V_2, .. , V_k$ such that for any two vertices $u \in V_i, v \in V_j$, there exists a directed cycle in $G$ containing $u$ and $v$ if and only if $i=j$. In the Single-Source Reachability (SSR) problem, we are given a distinguished source $r \in V$ and are asked to find all vertices in $V$ that can be reached from $r$. The SSR problem can be reduced to finding the SCCs by inserting edges from each vertex in $V$ to the distinguished source $r$.

Finding SCCs in static graphs in $O(m+n)$ time is well-known since 1972\cite{tarjan1972depth} and is commonly taught in undergraduate courses, also appearing in CLRS\cite{cormen2009introduction}. 

In this paper we focus on maintaining SCCs in a dynamic graph. The most general setting is the fully dynamic one, where edges are being inserted and deleted into the graph. While many connectivity problems for undirected graphs have been solved quite efficiently \cite{holm2001poly, wulff2013faster, thorup2000near, thorup2007fully, huang2017fully, nanongkai2017dynamic}, in fully-dynamic graphs, the directed versions of these problems have proven to be much harder to approach. 

In fact, Abboud and Vassilevska\cite{abboud2014popular} showed that any algorithm that can maintain whether there are more than 2 SCCs in a fully-dynamic graph with update time $O(m^{1-\epsilon})$ and query time $O(m^{1-\epsilon})$, for any constant $\epsilon > 0$, would imply a major breakthrough on SETH. The same paper also suggests that $O(m^{1-\epsilon})$ update time and query time $O(n^{1-\epsilon})$ for maintaining the number of reachable vertices from a fixed source would imply a breakthrough for combinatorial Matrix Multiplication. 

For this reason, research on dynamic SCC and dynamic single-source reachability has focused on the partially dynamic setting (decremental or incremental). In this paper we study the \emph{decremental} setting, where the original graph only undergoes edge deletions (no insertions). We note that both lower bounds above extend to decremental algorithms with \emph{worst-case} update time $O(m^{1-\epsilon})$, so all existing results focus on the amortized update time.

The first algorithm to maintain SSR faster than recomputation from scratch achieved total update time $O(mn)$\cite{shiloach1981line}. The same update time for maintaining SCCs was achieved by a randomized algorithm by Roddity and Zwick\cite{roditty2008improved}. Their algorithm also establishes that any algorithm for maintaining SSR can be turned into a randomized algorithm to maintain the SCCs incurring only an additional constant multiplicative factor in the running time. Later, Łącki\cite{lkacki2013improved} presented a simple deterministic algorithm that matches $O(mn)$ total update time and that also maintains the transitive closure. 
For several decades, it was not known how to get beyond total update time $O(mn)$, until a recent breakthrough by Henzinger, Krinninger and Nanongkai\cite{henzinger2014sublinear, henzinger2015improved} reduced the total update time to expected time $O(\min(m^{7/6}n^{2/3}, m^{3/4}n^{5/4+o(1)}, m^{2/3}n^{4/3+o(1)} + m^{3/7}n^{12/7+o(1)})) = O(mn^{0.9+o(1)})$. Even more recently, Chechik et. al.\cite{chechik2016decremental} showed that a clever combination of the algorithms of Roditty and Zwick, and Łącki can be used to improve the expected total update time to $\tilde{O}(m \sqrt{n})$. We point out that all of these recent results rely on randomization and in fact no deterministic algorithm for maintaining SCCs or SSR beyond the $O(mn)$ bound is known for general graphs. For planar graphs, Italiano et. al.\cite{italiano2017decremental} presented a deterministic algorithm with total update time $\tilde{O}(n)$. 

Finally, in this paper, we present the first algorithm for general graphs to maintain SCCs in $\tilde{O}(m)$ expected total update time with constant query time, thus presenting the first near-optimal algorithm for the problem. We summarize our result in the following theorem.

\begin{theorem}
\label{thm:SCCmain}
    Given a graph $G=(V,E)$ with $m$ edges and $n$ vertices, we can maintain a data structure that supports the operations:
    \begin{itemize}
        \item $\textsc{Delete}(u,v)$: Deletes the edge $(u,v)$ from the graph $G$,
        \item $\textsc{Query}(u,v)$: Returns whether $u$ and $v$ are in the same SCC in $G$,
    \end{itemize}
    in total expected update time $O(m \log^4 n)$ and with worst-case constant query time. The same time bounds apply to answer for a fixed source vertex $s \in V$ queries on whether a vertex $v \in V$ can be reached from $s$. The bound holds against an oblivious adaptive adversary.
\end{theorem}

Our algorithm makes the standard assumption of an oblivious adversary which does not have access to the coin flips made by the algorithm. But our algorithm does NOT require the assumption of a non-adaptive adversary, which is ignorant of answers to queries as well: the reason is simply that SCC and SSR information is unique, so the answers to queries do not reveal any information about the algorithm. One key exception is that for SSR, if the algorithm is expected to return a witness path, then it does require the assumption of a non-adaptive adversary. 

A standard reduction described in appendix \ref{sec:fullyReach} also implies a simple algorithm for maintaining reachability from some set $S \subseteq V$ to $V$ in a fully-dynamic graph with vertex set $V$ that is a data structure that answers queries for any $s \in S, v \in V$ on whether $s$ can reach $v$. The amortized expected update time is $\tilde{O}(|S|m/t)$ and query time $O(t)$ for every $t \in [1, |S|]$. We allow vertex updates, i.e. insertions or deletions of vertices with incident edges, which are more general than edge updates. This generalizes a well-known trade-off result for All-Pairs Reachability\cite{roditty2016fully, lkacki2013improved} with $\tilde{O}(nm/t)$ amortized update time and query time $O(t)$ for every $t \in [1, n]$.

Finally, we point out that maintaining SCCs and SSR is related to the more difficult (approximate) shortest-path problems. In fact, the algorithms \cite{shiloach1981line, henzinger1995fully, henzinger2014sublinear,henzinger2015improved} can also maintain (approximate) shortest-paths in decremental directed graphs. For undirected graphs, the decremental Single-Source Approximate Shortest-Path problem was recently solved to near-optimality\cite{henzinger2014decremental}, and deterministic algorithms\cite{bernstein2016deterministic, bernstein2017deterministic} have been developed that go beyond the $O(mn)$ barrier. We hope that our result inspires new algorithms to tackle the directed versions of these problems. 

\section{Preliminaries}
\label{sec:prelim}

In this paper, we let a graph $H = (V,E)$ refer to a directed multi-graph where we allow multiple edges between two endpoints and self-loops but say that a cycle contains at least two distinct vertices. We refer to the vertex set of $H$ by $V(H)$ and the edge set by $E(H)$. We denote the input graph by $G$, let  $V= V(G)$ and $E = E(G)$ and define $n = |V|$ and $m = |E|$. If the context is clear, we simply write sets $X$ instead of their cardinality $|X|$ in calculations to avoid cluttering. 

We define a subgraph of $H$ to be a graph $H'$ with $V(H') = V(H)$ and $E(H') \subseteq E(H)$. Observe that this deviates from the standard definition of subgraphs since we require the vertex set to be equivalent. We write $H \setminus E'$ as a shorthand for the graph $(V(H), E(H) \setminus E')$ and $H \cup E'$ as a shorthand for $(V(H), E(H) \cup E')$. For any $S \subseteq V(H)$, we define $E^H_{out}(S)$ to be the set $(S \times V(H)) \cap E(H)$, i.e. the set of all edges in $H$ that emanate from a vertex in $S$; we analogously define $E^H_{in}(S)$ and $E^H(S) = E^H_{in}(S) \cup E^H_{out}(S)$. If the context is clear, we drop the superscript and simply write $E_{in}(S), E_{out}(S), E(S)$.

For any graph $H$, and any two vertices $u,v \in V(H)$, we denote by $\mathbf{dist}_H(u,v)$ the distance from $u$ to $v$ in $H$. We also define the notion of $S$-distances for any $S \subseteq V(H)$ where for any pair of vertices $u,v \in V(H)$, the $S$-distance $\mathbf{dist}_H(u,v, S)$ denotes the minimum number of vertices in $S \setminus \{v\}$ encountered on any path from $u$ to $v$. Alternatively, the $S$-distance corresponds to $\mathbf{dist}_{H'}(u,v)$ where $H'$ is a graph with edges $E_{out}(S)$ of weight $1$ and edges $E \setminus E_{out}(S)$ of weight $0$. It therefore follows that for any $u,v \in V(H)$, $\mathbf{dist}_H(u,v) = \mathbf{dist}_H(u,v,V)$. 

We define the diameter of a graph $H$ by $\mathbf{diam}(H) = \max_{u,v \in V} \mathbf{dist}_H(u,v)$ and the $S$-diameter by $\mathbf{diam}(H, S) = \max_{u,v \in V(H)} \mathbf{dist}_H(u,v,S)$. Therefore, $\mathbf{diam}(H) = \mathbf{diam}(H, V)$. For convenience, we often omit the subscript on relations if the context is clear and write $\mathbf{dist}(u,v,S)$ 

We denote that a vertex $u$ \textit{reaches} $v$ in $H$ by $u \leadsto_H v$, and if $u \leadsto_H v$ and $v \leadsto_H u$, we simply write $u \rightleftarrows_H v$ and say $u$ and $v$ are \textit{strongly-connected}. We also use $\leadsto$ and $\rightleftarrows$ without the subscript if the underlying graph $H$ is clear from the context. We say that $H$ is strongly-connected if for any $u,v \in V(H)$, $u \rightleftarrows v$. We call the maximal subgraphs of $H$ that are strongly-connected, the strongly-connected components (SCCs). We denote by $\textsc{Condensation}(H)$ the \textit{condensation} of $H$, that is the graph where all vertices in the same SCC in $H$ are contracted. To distinguish we normally refer to the vertices in $\textsc{Condensation}(H)$ as \textit{nodes}. Each node in $\textsc{Condensation}(H)$ corresponds to a vertex set in $H$. The node set of a condensation $\textsc{Condensation}(H)$ forms a partition of $V(H)$. For convenience we define the function $\textsc{Flatten}(X)$ for a family of sets $X$ with $\textsc{Flatten}(X) = \bigcup_{x \in X} x$. This is useful when discussing condensations. Observe further that $\textsc{Condensation}(H)$ can be a multi-graph and might also contain self-loops. If we have an edge set $E'$ with all endpoints in $H$, we let $\textsc{Condensation}(H) \cup E'$ be the multi-graph obtained by mapping the endpoints of each vertex in $E'$ to their corresponding SCC node in $\textsc{Condensation}(H)$ and adding the resulting edges to $\textsc{Condensation}(H)$.

Finally, for two partitions $P$ and $P'$ of a set $U$, we say that partition $P$ is a \textit{melding} for a partition $P'$ if for every set $X \in P'$, there exists a set $Y \in P$ with $X \subseteq Y$. We also observe that \textit{melding} is transitive, thus if $P$ is a melding for $P'$ and $P'$ a melding for $P''$ then $P$ is a melding for $P''$.


\section{Overview}
\label{subsec:overview}


We now introduce the graph hierarchy maintained by our algorithm, followed by a high-level overview of our algorithm.

\paragraph{High-level overview of the hierarchy.}  Our hierarchy has levels $0$ to $\lfloor \lg n \rfloor + 1$ and we associate with each level $i$ a subset $E_i$ of the edges $E$. 
The sets $E_i$ form a partition of $E$; we define the edges that go into each $E_i$ later in the overview but point out that we maintain $E_{\lfloor \lg n \rfloor + 1} = \emptyset$. We define a graph hierarchy $\hat{G} = \{\hat{G}_0, \hat{G}_1, .. ,\hat{G}_{\lfloor \lg n \rfloor + 1}\}$ such that each graph $\hat{G}_i$ is defined as 
\[
\hat{G}_i = \textsc{Condensation}((V, \bigcup_{j < i} E_j)) \cup E_{i}
\]
That is, each $\hat{G}_i$ is the condensation of a subgraph of $G$ with some additional edges. As mentioned in the preliminary section, we refer to the elements of the set $\hat{V}_i = V(\hat{G}_i)$ as \textit{nodes} to distinguish them from \textit{vertices} in $V$. We use capital letters to denote nodes and small letters to denote vertices. We let $X_i^v$ denote the node in $\hat{V}_i$ with $v \in X_i^v$. Observe that each node $X$ corresponds to a subset of vertices in $V$ and that for any $i$, $\hat{V}_i$ can in fact be seen as a partition of $V$. For $\hat{G}_0 = \textsc{Condensation}((V, \emptyset)) \cup E_{0}$, the set $\hat{V}_0$ is a partition of singletons, i.e. $\hat{V}_0 = \{ \{v\} | v \in V\}$, and $X_0^v = \{ v\}$ for each $v \in V$. 

Observe that because the sets $E_i$ form a partition of $E$ and $E_{\lfloor \lg n \rfloor + 1} = \emptyset$, the top graph $\hat{G}_{\lfloor \lg n \rfloor + 1}$ is simply defined as
\[
\hat{G}_{\lfloor \lg n \rfloor + 1} = \textsc{Condensation}((V, \bigcup_{j < \lfloor \lg n \rfloor + 1} E_j)) \cup E_{\lfloor \lg n \rfloor + 1} = \textsc{Condensation}((V,E)).
\]
Therefore, if we can maintain $\hat{G}_{\lfloor \lg n \rfloor + 1} $ efficiently, we can answer queries on whether two vertices $u,v \in V$ are in the same SCC in $G$ by checking if $X_{\lfloor \lg n \rfloor + 1}^u$ is equal to $X_{\lfloor \lg n \rfloor + 1}^v$.

Let us offer some intuition for the hierarchy. The graph $\hat{G}_0$ contains all the vertices of $G$, and all the edges of $E_0 \subseteq E$. By definition of $\textsc{Condensation}(\cdot)$, the nodes of $\hat{G}_1$ precisely correspond to the SCCs of $\hat{G}_0$. $\hat{G}_1$ also includes the edges $E_0$ (though some of them are contracted into self-loops in $\textsc{Condensation}((V, E_0))$), as well as 
the additional edges in $E_1$. These additional edges might lead to $\hat{G}_1$ having larger SCCs than those of $\hat{G}_0$; each SCC in $\hat{G}_1$ then corresponds to a node in $\hat{G_2}$. More generally, the nodes of $\hat{G}_{i+1}$ are the SCCs of $\hat{G}_{i}$. 

As we move up the hierarchy, we add more and more edges to the graph, so the SCCs get larger and larger. Thus, each set $\hat{V}_i$ is a \textit{melding} for any $\hat{V}_j$ for $j \leq i$; that is for each node $Y \in \hat{V}_j$ there exists a set $X \in \hat{V}_i$ such that $Y \subseteq X$. We sometimes say we \textit{meld} nodes $Y, Y' \in \hat{V}_j$ to $X \in \hat{V}_{i}$ if $Y, Y' \subseteq X$ and $j < i$. Additionally, we observe that for any SCC $Y \subseteq \hat{V}_i$ in $\hat{G}_i$, we meld the nodes in SCC $Y$ to a node in $X \in \hat{V}_{i+1}$, and $X$ consists exactly of the vertices contained in the nodes of $Y$. More formally, $X = \textsc{Flatten}(Y)$.


To maintain the SCCs in each graph $\hat{G}_i$, our algorithm employs a bottom-up approach. At level $i+1$ we want to maintain SCCs in the graph with all the edges in $\bigcup_{j \leq i+1} E_j$, but instead of doing so from scratch, we use the SCCs maintained at level $\hat{G}_i$ as a starting point. The SCCs in $\hat{G}_i$ are precisely the SCCs in the graph with edge set $\bigcup_{j \leq i} E_j$; so to maintain the SCCs at level $i+1$, we only need to consider how the sliver of edges in $E_{i+1}$ cause the SCCs in $\hat{G}_{i}$ to be melded into larger SCCs (which then become the nodes of $\hat{G}_{i+2}$).

If the adversary deletes an edge in $E_i$, all the graphs $\hat{G}_{i-1}$ and below remain unchanged, as do the nodes of $\hat{G}_{i}$. But the deletion might split apart an SCC in $G_i$, which will in turn cause a node of $\hat{G}_{i+1}$ to split into multiple nodes. This split might then cause an SCC of $\hat{G}_{i+1}$ to split, which will further propagate up the hierarchy.

In addition to edge deletions caused by the adversary, our algorithm will sometimes move edges from $E_i$ to $E_{i+1}$. Because the algorithm only moves edges \emph{up} the hierarchy, each graph $\hat{G}_i$ is only losing edges, so the update sequence remains decremental from the perspective of each $\hat{G}_i$. We now give an overview of how our algorithm maintains the hierarchy efficiently.


\paragraph{ES-trees.} A fundamental data structure that our algorithm employs is the ES-tree\cite{shiloach1981line, henzinger1995fully} that for a directed unweighted graph $G=(V,E)$ undergoing edge deletions, and a distinguished source $r \in V$ maintains the distance $\mathbf{dist}_G(r, v)$ for each $v \in V$. In fact, the ES-tree maintains a shortest-path tree rooted at $r$. We refer to this tree subsequently as ES out-tree. We call the ES in-tree rooted at $r$ the shortest-path tree maintained by running the ES-tree data structure on the graph $G$ with reversed edge set, i.e. the edge set where each edge $(u,v) \in E$ appears in the form $(v,u)$. We can maintain each in-tree and out-tree decrementally to depth $\delta > 0$ in time $O(|E| * \delta)$; that is we can maintain the distances $\mathbf{dist}_G(r, v)$ and $\mathbf{dist}_G(v, r)$ exactly
until one of the distances $\mathbf{dist}_G(r, v)$ or $\mathbf{dist}_G(v, r)$ exceeds $\delta$. 

\paragraph{Maintaining SCCs with ES-trees.} Consider again graph $\hat{G}_0$ and let $X \subseteq \hat{V}_0$ be some SCC in $\hat{G}_0$ that we want to maintain. Let some node $X'$ in $X$ be chosen to be the \textit{center} node of the SCC (In the case of $\hat{G}_0$, the node $X'$ is just a single-vertex set $\{ v \}$). We then maintain an ES in-tree and an ES out-tree from $X'$ that spans the nodes in $X$ in the induced graph $\hat{G}_0[X]$. We must maintain the trees up to distance $\mathbf{diam}(\hat{G}_0[X])$, so the total update time is $O(|E(\hat{G}_0[X])|* \mathbf{diam}(\hat{G}_0[X]))$. 

Now, consider an edge deletion to $\hat{G}_0$ such that the ES in-tree or ES out-tree at $X'$ is no longer a spanning tree. Then, we detected that the SCC $X$ has to be split into at least two SCCs $X_1, X_2, .., X_k$ that are node-disjoint with $X = \bigcup_i X_i$. Then in each new SCC $X_i$ we choose a new center and initialize a new ES in-tree and ES out-tree.

\paragraph{Exploiting small diameter.} The above scheme clearly is quite efficient if $\mathbf{diam}(\hat{G}_0)$ is very small. Our goal is therefore to choose the edge set $E_0$ in such a way that $\hat{G}_0$ contains only SCCs of small diameter. We therefore turn to some insights from \cite{chechik2016decremental} and extract information from the ES in-tree and out-tree to maintain small diameter. Their scheme fixes some $\delta > 0$ and if a set of nodes $Y \subseteq X$ for some SCC $X$ is at distance $\Omega(\delta)$ from/to $\textsc{Center}(X)$ due to an edge deletion in $\hat{G}_0$, they find a node separator $S$ of size $O(\min\{|Y|, |X \setminus Y|\}\log n / \delta)$; removing $S$ from  $\hat{G}_0$ causes $Y$ and $X \setminus Y$ to no longer be in the same SCC. We use this technique and remove edges incident to the node separator $S$ from $E_0$ and therefore from $\hat{G}_0$. One subtle observation we want to stress at this point is that each node in the separator set appears also as a single-vertex node in the graph $\hat{G}_1$; this is because each separator node $\{ s\}$ for some $s \in V$ is not \textit{melded} with any other node in $\hat{V}_0$, as it has no edges in $\hat{G}_0$ to or from any other node.

For some carefully chosen $\delta = \Theta(\log^2 n)$, we can maintain $\hat{G}_0$ such that at most half the nodes in $\hat{V}_0$ become separator nodes at any point of the algorithm.  This follows since each separator set is small in comparison to the smaller side of the cut and since each node in $\hat{V}_0$ can only be $O(\log n)$ times on the smaller side of a cut. 

\paragraph{Reusing ES-trees.} Let us now refine our approach to maintain the ES in-trees and ES out-trees and introduce a crucial ingredient devised by Roditty and Zwick\cite{roditty2008improved}. Instead of picking an arbitrary center node $X'$ from an SCC $X$ with $X' \in X$, we are going to pick a vertex $r \in \textsc{Flatten}(X) \subseteq V$ uniformly at random and run our ES in-tree and out-tree $\mathcal{E}_{r}$ from the node $X_0^r$ on the graph $\hat{G}_0$. For each SCC $X$ we denote the randomly chosen root $r$ by $\textsc{Center}(X)$. In order to improve the running time, we \textit{reuse} ES-trees when the SCC $X$ is split into SCCs $X_1, X_2, .. , X_k$, where we assume wlog that $r \in \textsc{Flatten}(X_1)$, by removing the nodes in $X_2, .., X_k$ from $\mathcal{E}_{r}$ and setting $\textsc{Center}(X_1) = r$. Thus, we only need to initialize a new ES-tree for the SCC $X_2, .. , X_k$. Using this technique, we can show that each node is expected to participate in $O(\log n)$ ES-trees over the entire course of the algorithm, since we expect that if a SCC $X$ breaks into SCCs $X_1, X_2, .. , X_k$ then we either have that every SCC $X_i$ is of at most half the size of $X$, or with probability at least $1/2$ that $X_1$ is the new SCC that contains at least half the vertices, i.e. that the random root is containing in the largest part of the graph. Since the ES-trees work on induced graphs with disjoint node sets, we can therefore conclude that the total update time for all ES-trees is $O(m \log n* \mathbf{diam}(\hat{G}_0))$. 

We point out that using the ES in-trees and out-trees to detect node separators as described above complicates the analysis of the technique by Roditty and Zwick\cite{roditty2008improved} but a clever proof presented in \cite{chechik2016decremental} shows that the technique can still be applied. In our paper, we present a proof that can even deal with some additional complications and that is slightly simpler.

\paragraph{A contrast to the algorithm of Chechik et al \cite{chechik2016decremental}.}

Other than our hierarchy, the overview we have given so far largely comes from the algorithm of Chechik et al \cite{chechik2016decremental}. However, their algorithm does not use a hierarchy of graphs. Instead, they show that for any graph $G$, one can find (and maintain) a node separator $S$ of size $\tilde{O}(n/\delta)$ such that all SCCs in $G$ have diameter at most $\delta$. They can then use ES-trees with random sources to maintain the SCCs in $G \setminus S$ in total update time $\tilde{O}(m\delta)$. This leaves them with the task of computing how the vertices in $S$ might meld some of the SCCs in $G \setminus S$. They are able to do this in total update time $\tilde{O}(m|S|) = \tilde{O}(mn/\delta)$ by using an entirely different technique of \cite{lkacki2013improved}. Setting $\delta = \tilde{O}(\sqrt{n})$, they achieve the optimal trade-off between the two techniques: total update time $\tilde{O}(m\sqrt{n})$ in expectation.

We achieve our $\tilde{O}(m)$ total update time by entirely avoiding the technique of \cite{lkacki2013improved} for separately handling a small set of separator nodes, and instead using the graph hierarchy described above, where at each level we set $\delta$ to be polylog rather than $\tilde{O}(\sqrt{n})$. 

We note that while our starting point is the same as \cite{chechik2016decremental}, using a hierarchy of separators forces us to take a different perspective on the function of a separator set. The reason is that it is simply not possible to ensure that at each level of the hierarchy, all SCCs have small diameter. To overcome this, we instead aim for separator sets that decompose the graph into SCCs that are small with respect to a different notion of distance. The rest of the overview briefly sketches this new perspective, while sweeping many additional technical challenges under the rug. 

\paragraph{Refining the hierarchy.} So far, we only discussed how to maintain $\hat{G}_0$ efficiently by deleting many edges from $E_0$ and hence ensuring that SCCs in $\hat{G}_0$ have small diameter. To discuss our bottom-up approach, let us define our graphs $\hat{G}_i$ more precisely.

We maintain a separator hierarchy $\mathcal{S} = \{S_0, S_1, .. , S_{\lfloor \lg{n} \rfloor + 2}\}$ where $\hat{V}_0 = S_0 \supseteq S_1 \supseteq .. \supseteq S_{\lfloor \lg{n} \rfloor + 1} = S_{\lfloor \lg{n} \rfloor + 2} = \emptyset$, with $|S_i| \leq n/2^i$, for all $i \in [0, \lfloor \lg{n} \rfloor + 2]$ (see below that for technical reasons we need to define $S_{\lfloor \lg{n} \rfloor + 2}$ to define $\hat{G}_{\lfloor \lg{n} \rfloor + 1}$). Each set $S_i$ is a set of single-vertex nodes -- i.e. nodes of the form $\{ v \}$ -- that is monotonically increasing over time. 

We can now more precisely define each edge set $E_i = E(\textsc{Flatten}(S_{i} \setminus S_{i+1}))$. To avoid clutter, we abuse notation slightly referring henceforth to $\textsc{Flatten}(X)$ simply as $X$ if $X$ is a set of singleton sets and the context is clear. We therefore obtain
\[
\hat{G}_i =  \textsc{Condensation}((V, \bigcup_{j < i} E_j)) \cup E_{i} = \textsc{Condensation}(G \setminus E(S_{i})) \cup E(S_{i} \setminus S_{i+1})).
\]
In particular, note that $\hat{G}_i$ contains all the edges of $G$ except those in $E(S_{i+1})$; as we move up to level $\hat{G}_{i+1}$, we add the edges incident to $S_{i+1} \setminus S_{i+2}$. Note that if $s \in S_{i} \setminus S_{i+1}$, and our algorithm then adds $s$ to $S_{i+1}$, this will remove all edges incident to $s$ from $E_{i}$ and add them to $E_{i+1}$. Thus the fact that the sets $S_i$ used by the algorithm are monotonically increasing implies the desired property that edges only move up the hierarchy (remember that we add more vertices to $S_i$ due to new separators found on level $i-1$).

At a high-level, the idea of the hierarchy is as follows. Focusing on a level $i$, when the ``distances'' in some SCC of $\hat{G}_i$ get too large (for a notion of distance defined below), the algorithm will add a carefully chosen set of separator nodes $s_1, s_2, ..$ in $S_i$ to $S_{i+1}$. By definition of our hierarchy, this will remove the edges incident to the $s_i$ from $\hat{G}_i$, thus causing the SCCs of $\hat{G}_i$ to decompose into smaller SCCs with more manageable ``distances''. We note that our algorithm always maintains the invariant that nodes added to $S_{i+1}$ were previously in $S_i$, which from the definition of our hierarchy, ensures that at all times the separator nodes in $S_{i+1}$ are single-vertex nodes in $\hat{V}_{i+1}$; this is because the nodes of $\hat{V}_{i+1}$ are the SCCs of $\hat{G}_i$, and $\hat{G}_i$ contains no edges incident to $S_{i+1}$.


\paragraph{Exploiting $S$-distances.} For our algorithm, classic ES-trees are only useful to maintain SCCs in $\hat{G}_0$; in order to handle levels $i > 0$ we develop a new generalization of ES-trees that use a different notion of distance. This enables us to detect when SCCs are split in graphs $\hat{G}_i$ and to find separator nodes in $\hat{G}_i$ as discussed above more efficiently. 

Our generalized ES-tree (GES-tree) can be seen as a combination of the classic ES-trees\cite{shiloach1981line} and a data structure by Italiano\cite{italiano1988finding} that maintains reachability from a distinguished source in a directed acyclic graph (DAG), and which can be implemented in total update time $O(m)$.

Let $S$ be some feedback vertex set in a graph $G = (V,E)$; that is, every cycle in $G$ contains a vertex in $S$. Then our GES-tree can maintain
$S$-distances and a corresponding shortest-path tree up to $S$-distance $\delta > 0$ from a distinguished source $X_i^r$ for some $r \in V$ in the graph $G$. (See Section \ref{sec:prelim} for the definition of $S$-distances.) This data structure can be implemented to take $O(m\delta)$ total update time.


\paragraph{Maintaining the SCCs in $\hat{G}_i$.} Let us focus on maintaining SCCs in $\hat{G}_i = \textsc{Condensation}(G \setminus E(S_{i})) \cup E(S_{i} \setminus S_{i+1})$. Since the condensation of any graph forms a DAG, every cycle in $\hat{G}_i$ contains at least one edge from the set $E(S_{i} \setminus S_{i+1})$. Since $E(S_{i} \setminus S_{i+1})$ is a set of edges that is incident to $S_i$, we have that $S_i$ forms a feedback node set of $\hat{G}_i$. Now consider the scheme described in the paragraphs above, but instead of running an ES in-tree and out-tree from each center $\textsc{Center}(X)$ for some SCC $X$, we run a GES in-tree and out-tree on $\hat{G}_i[X]$ that maintains the $S_i$-distances to depth $\delta$.  Using this GES, whenever a set $Y \subseteq X$ of nodes has $S_i$-distance $\Omega(\delta)$, we show that we can find a separator $S$ of size $O(\min\{|S_i \cap Y|, |S_i \cap (X \setminus Y)|\}\log n / \delta)$ that only consists of nodes that are in $\{\{s\} | s \in S_i\}$; we then add the elements of set $S$ to the set $S_{i+1}$, and we also remove the nodes $Y$ from the GES-tree, analogously to our discussion of regular ES-trees above. Note that adding $S$ to $S_{i+1}$ removes the edges $E(S)$ from $\hat{G}_i$; since we chose $S$ to be a separator, this causes $Y$ and $X \setminus Y$ to no longer be part of the same SCC in $\hat{G}_i$. Thus, to maintain the hierarchy, we must then split nodes in $\hat{G}_{i+1}$ into multiple nodes corresponding to the new SCCs in $\hat{G}_i$: $X \setminus (Y \cup S)$, $Y \setminus S$ and every single-vertex set in $S$ ($Y$ might not form a SCC but we then further decompose it after we handled the node split). This might cause some self-loops in $\hat{G}_{i+1}$ to become edges between the newly inserted nodes (resulting from the split) and needs to be handled carefully to embed the new nodes in the GES-trees maintained upon the SCC in $\hat{G}_{i+1}$ that $X$ is part of. Observe that this does not result in edge insertions but only remaps the endpoints of edges. Further observe that splitting nodes can only increase $S_{i+1}$-distance since when they were still contracted their distance from the center was equivalent. Since $S_{i+1}$-distance still might increase, the update to might trigger further changes in the graph $\hat{G}_{i+1}$. 

Thus, overall, we ensure that all SCCs in $\hat{G}_i$ have $S_i$-diameter at most $O(\delta)$, and can hence be efficiently maintained by GES-trees. In particular, we show that whenever an SCC exceeds diameter $\delta$, we can, by moving a carefully chosen set of nodes in $S_i$ to $S_{i+1}$, remove a corresponding set of edges in $\hat{G}_i$, which breaks the large-$S_i$-diameter SCC into SCCs of smaller $S_i$-diameter.

\paragraph{Bounding the total update time.} Finally, let us sketch how to obtain the total expected running time $O(m \log^4 n)$. We already discussed how by using random sources in GES-trees (analogously to the same strategy for ES-trees), we ensure that each node is expected to be in $O(\log n)$ GES-trees maintained to depth $\delta = O(\log^2 n)$. 
Each such GES-tree is maintained in total update time $O(m\delta) = O(m\log^2n)$, so we have $O(m \log^3 n)$ total expected update time for each level, and since we have $O(\log n)$ levels, we obtain total expected update time $O(m\log^4 n)$. We point out that we have not included the time to compute the separators in our running time analysis; indeed, computing separators efficiently is one of the major challenges to building our hierarchy. Since implementing these subprocedures efficiently is rather technical and cumbersome, we omit their description from the overview but refer to section \ref{subsec:separators} for a detailed discussion.

\section{Generalized ES-trees}
\label{subsec:EStree}

Even and Shiloach\cite{shiloach1981line} devised a data structure commonly referred to as ES-trees that given a vertex $r \in V$ in a graph $G=(V,E)$ undergoing edge deletions maintains the shortest-path tree from $r$ to depth $\delta$ in total update time $O(m \delta)$ such that the distance $\mathbf{dist}_G(r,v)$ of any vertex $v \in V$ can be obtained in constant time. Henzinger and King\cite{henzinger1995fully} later observed that the ES-tree can be adapted to maintain the shortest-path tree in directed graphs. 

For our algorithm, we devise a new version of the ES-trees that maintains the shortest-path tree with regard to $S$-distances. We show that if $S$ is a \textit{feedback vertex set} for $G$, that is a set such that every cycle in $G$ contains at least one vertex in $S$, then the data structure requires only $O(m\delta)$ total update time. Our fundamental idea is to combine classic ES-trees with techniques to maintain Single-Source Reachability in DAGs which can be implemented in linear time in the number of edges\cite{italiano1988finding}. Since $\mathbf{dist}_G(r, v) = \mathbf{dist}_G(r, v, V)$  and $V$ is a trivial feedback vertex set, we have that our data structure generalizes the classic ES-tree. Since the empty set is a feedback vertex set for DAGs, our data structure also matches the time complexity of Italiano's data structure. We define the interface formally below.

\begin{definition}
\label{def:GES}
Let $G=(V,E)$ be a graph and $S$ a feedback vertex set for $G$, $r \in V$ and $\delta > 0$. We define a generalized ES-tree $\mathcal{E}_r$ (GES) to be a data structure that supports the following operations:
\begin{itemize}
    \item $\textsc{InitGES}(r, G, S, \delta)$: Sets the parameters for our data structure. We initialize the data structure and return the GES.
    \item $\textsc{Distance}(r ,v )$: $\forall v \in V$, if $\mathbf{dist}_G(r ,v, S) \leq \delta$, $\mathcal{E}_r$ reports $\mathbf{dist}_G(r ,v, S)$, otherwise $\infty$.
    \item $\textsc{Distance}(v ,r )$: $\forall v \in V$, if $\mathbf{dist}_G(v ,r, S) \leq \delta$, $\mathcal{E}_r$ reports $\mathbf{dist}_G(v ,r, S)$, otherwise $\infty$.
    \item $\textsc{Delete}(u,v)$: Sets $E \gets E \setminus \{(u,v)\}$.
    \item $\textsc{Delete}(V')$: For $V' \subseteq V$, sets $V \gets V \setminus V'$, i.e. removes the vertices in $V'$ and all incident edges from the graph $G$.
    \item $\textsc{GetUnreachableVertex}()$: Returns a vertex $v \in V$ with $\max\{ \mathbf{dist}_G(r ,v, S), \mathbf{dist}_G(v, r, S)\} > \delta$ or $\bot$ if no such vertex exists.
\end{itemize}
\end{definition}

\begin{lemma}
\label{lma:SimpleGES}
The GES $\mathcal{E}_r$ as described in definition \ref{def:GES} can be implemented with total initialization and update time $O(m \delta)$ and requires worst-case time $O(1)$ for each operation $\textsc{Distance}(\cdot)$ and $\textsc{GetUnreachableVertex}()$.
\end{lemma}

We defer the full prove to the appendix \ref{sec:proofsimpleges} but sketch the proof idea. 

\begin{proof} (sketch)
Consider a classic ES-tree with each edge weight $w(u,v)$ of an edge $(u,v)$ in $E_{out}(S)$ set to $1$ and all other edges of weight $0$. Then, the classic ES-tree analysis maintains with each vertex $v \in V$ the distance level $l(v)$ that expresses the current distance from $s$ to $v$. We also have a shortest-path tree $T$, where the path in the tree from $s$ to $v$ is of weight $l(v)$. Since $T$ is a shortest-path tree, we also have that for every edge $(u,v) \in E$, $l(v) \leq l(u) + w(u,v)$. Now, consider the deletion of an edge $(u,v)$ from $G$ that removes an edge that was in $T$. To certify that the level $l(v)$ does not have to be increased, we scan the in-going edges at $v$ and try to find an edge $(u', v) \in E$ such that $l(v) = l(u') + w(u', v)$. On finding this edge, $(u',v)$ is added to $T$. The problem is that if we allow $0$-weight cycles, the edge $(u',v)$ that we use to reconnect $v$ might come from a $u'$ that was a descendant of $v$ in $T$. This will break the algorithm, as it disconnects $v$ from $s$ in $T$. But we show that this bad case cannot occur because $S$ is assumed to be a feedback vertex set, so at least one of the vertices on the cycle must be in $S$ and therefore the out-going edge of this vertex must have weight $1$ contradicting that there exists any $0$-weight cycle. The rest of the analysis follows closely the classic ES-tree analysis.
\end{proof}

To ease the description of our SCC algorithm, we tweak our GES implementation to work on the multi-graphs $\hat{G}_i$. We still root the GES at a vertex $r \in V$, but maintain the tree in $\hat{G}_i$ at $X_i^r$. The additional operations and their running time are described in the following lemma whose proof is straight-forward and therefore deferred to appendix \ref{sec:proofAugmentedGES}. Note that we now deal with nodes rather than vertices which makes the definition of $S$-distances ambitious (consider for example a node containing two vertices in $S$). For this reason, we require $S \subseteq \{\{v\} | v \in V\} \cap \hat{V}$ in the lemma below, i.e. that the nodes containing vertices in $S$ are single-vertex nodes. But as discussed in the paragraph ``Refining the hierarchy" in the overview, our hierarchy ensures that every separator node $X \in S_i$ is always just a single-vertex node in $\hat{G}_i$. Thus this constraint can be satisfied by our hierarchy.

\begin{lemma}
\label{lma:AugmentedGES}
Say we are given a partition $\hat{V}$ of a universe $V$ and the graph $\hat{G}=(\hat{V}, E)$, a feedback node set $S \subseteq \{\{v\} | v \in V\} \cap \hat{V}$, a distinguished vertex $r \in V$, and a positive integer $\delta$. Then, we can run a GES $\mathcal{E}_{r}$ as in definition \ref{def:GES} on $\hat{G}$ in time $O(m\delta + \sum_{X \in \hat{V}} E(X) \log X)$ supporting the additional operations:
\begin{itemize}
    \item $\textsc{SplitNode}(X)$: the input is a set of vertices $X$ contained in node $Y \in \hat{V}$, such that either $E \cap (X \times Y \setminus X)$ or $E \cap (Y \setminus X \times X)$ is an empty set, which implies $X \not\rightleftarrows_{\hat{G} \setminus S} Y \setminus X$. We remove the node $Y$ in $\hat{V}$ and add node $X$ and $Y \setminus X$ to $\hat{V}$.
    \item $\textsc{Augment}(S')$: This procedure adds the nodes in $S'$ to the feedback vertex set $S$. Formally,
    the input is a set of single-vertex sets $S' \subseteq \{\{v\} | v \in V\} \cap \hat{V}$. $\textsc{Augment}(S')$ then adds every $s \in S'$ to $S$.
\end{itemize}
\end{lemma}

We point out that we enforce the properties on the set $X$ in the operation $\textsc{SplitNode}(X)$ in order to ensure that the set $S$ remains a feedback node set at all times. 

\section{Initializing the \texorpdfstring{the graph hierarchy $\hat{\mathcal{G}}$}{the graph hierarchy}} 
\label{subsec:Preprocessing}

We assume henceforth that the graph $G$ initially is strongly-connected. If the graph is not strongly-connected, we can run Tarjan's algorithm\cite{tarjan1972depth} in $O(m+n)$ time to find the SCCs of $G$ and run our algorithm on each SCC separately. 

\begin{algorithm}
\caption{$\textsc{Preprocessing}(G, \delta)$}
\label{alg:preprocessing}
\KwIn{A strongly-connected graph $G=(V,E)$ and a parameter $\delta > 0$.}
\KwOut{A hierarchy of sets $\mathcal{S} = \{ S_0, S_1, .., S_{\lfloor \lg{n} \rfloor + 2}\}$ and graphs $\hat{\mathcal{G}} = \{\hat{G}_0, \hat{G}_1, .. , \hat{G}_{\lfloor \lg{n} \rfloor + 1}\}$ as described in section \ref{subsec:overview}. Further, each SCC $X$ in $\hat{G}_i$ for $i \leq \lfloor \lg n \rfloor + 2$, has a center $\textsc{Center}(X)$ such that for any $y \in X$, $\mathbf{dist}_{\hat{G}_{i}}(\textsc{Center}(X),y, S_{i}) \leq \delta/2$ and $\mathbf{dist}_{\hat{G}_{i}}(y,\textsc{Center}(X), S_{i}) \leq \delta/2$. }
\BlankLine

$ S_0 \gets V$\;
$ \hat{V}_0 \gets \{\{v\} | v \in V\}$\;
$ \hat{G}_0 \gets (\hat{V}_0, E)$\label{lne:G0Init}\;
\For{ $i = 0 $ \KwTo $ \lfloor \lg{n} \rfloor$}{
    \tcc{Find separator $S_{Sep}$ such that no two vertices in the same SCC in $\hat{G}_{i}$ have $S_i$-distance $\geq \delta/2$. $P$ is the collection of these SCCs.}
	$ (S_{Sep}, P) \gets \textsc{Split}(\hat{G}_i, S_{i}, \delta/2)$ \;
	$ S_{i+1} \gets S_{Sep} $\;
	$\textsc{InitNewPartition}(P, i, \delta)$\;
    
    \tcc{Initialize the graph $\hat{G}_{i+1}$}
	$ \hat{V}_{i+1} \gets P$\;
    $ \hat{G}_{i+1} \gets (\hat{V}_{i+1}, E)$\label{lne:GIInit}\;
}        
$S_{\lfloor \lg{n} \rfloor + 2} \gets \emptyset$\;
\end{algorithm}

Our procedure to initialize our data structure is presented in pseudo-code in algorithm \ref{alg:preprocessing}. We first initialize the level $0$ where $\hat{G}_0$ is simply $G$ with the vertex set $V$ mapped to the set of singletons of elements in $V$. 

Let us now focus on an iteration $i$. Observe that the graph $\hat{G}_i$ initially has all edges in $E$ (by initializing each $\hat{G}_i$ in line \ref{lne:G0Init} or line \ref{lne:GIInit}). Our goal is then to ensure that all SCCs in $\hat{G}_i$ are of small $S_i$-diameter at the cost of removing some of the edges from $\hat{G}_i$. Invoking the procedure $\textsc{Split}(\hat{G}_i, S_i, \delta/2)$ provides us with a set of separator nodes $S_{Sep}$ whose removal from $\hat{G}_i$ ensure that the $S_i$ diameter of all remaining SCCs is at most $\delta$. The set $P$ is the collection of all these SCCs, i.e. the collection of the SCCs in $\hat{G}_i \setminus E(S_{Sep})$. 

Lemma \ref{lma:split} below describes in detail the properties satisfied by the procedure $\textsc{Split}(\hat{G}_i, S_i, \delta/2)$. In particular, besides the properties ensuring small $S_i$-diameter in the graph $\hat{G}_i \setminus E(S_{Sep})$ (properties \ref{prop:1} and \ref{prop:2}), the procedure also gives an upper bound on the number of separator vertices (property \ref{prop:3}). Setting $\delta = 64 \lg^2 n$, clearly implies that $|S_{Sep}| \leq |S_i|/2$ and ensures running time $O(m \log^3 n)$.

\begin{lemma}
\label{lma:split}
$\textsc{Split}(\hat{G}_i, S_i, \delta/2)$ returns a tuple $(S_{Sep}, P)$ where $P$ is a partition of the node set $\hat{V}_i$ such that 
\begin{enumerate}
    \item for $X \in P$, and nodes $u,v \in X$ we have  $\mathbf{dist}_{G \setminus E(S_{Sep})}(u,v,S) \leq \delta/2$, and \label{prop:1}
    \item for distinct $X, Y \in P$, with nodes $u \in X$ and $v \in Y$,   $u \not\rightleftarrows_{G \setminus E(S_{Sep})} v$, and \label{prop:2}
    \item 
    $|S_{Sep}| \leq  \frac{32 \lg^2 n}{\delta} |S_i|$. \label{prop:3}
\end{enumerate}
The algorithm runs in time $O\left(\delta m \lg n\right)$.
\end{lemma}

We then set $S_{i+1} = S_{Sep}$ which implicitly removes the edges $E(S_{i+1})$ from the graph $\hat{G}_i$. We then invoke the procedure $\textsc{InitNewPartition}(P, i, \delta)$, that is presented in algorithm \ref{alg:newPart}. The procedure initializes for each $X \in P$ that corresponds to an SCC in $\hat{G}_i$ the GES-tree from a vertex $r \in \textsc{Flatten}(X)$ chosen uniformly at random on the induced graph $\hat{G}_i[X]$. Observe that we are not explicitly keeping track of the edge set $E_i$ but further remove edges implicitly by only maintaining the induced subgraphs of $\hat{G}_i$ that form SCCs. A small detail we want to point out is that each separator node $X \in S_{Sep}$ also forms its own single-node set in the partition $P$.

\begin{algorithm}
\caption{$\textsc{InitNewPartition}(P, i, \delta)$\;}
\label{alg:newPart}
\KwIn{A partition of a subset of the nodes $V$, and the level $i$ in the hierarchy.}
\KwResult{Initializes a new ES-tree for each set in the partition on the induced subgraph $\hat{G}_i$.}
\BlankLine
\ForEach{$ X \in P $}{
    Let $r$ be a vertex picked from $\textsc{Flatten}(X)$ uniformly at random. \;
    $\textsc{Center}(X) \gets r$\; 
    \tcc{Init a generalized ES-tree from $\textsc{Center}(X)$ to depth $\delta$.}
    $\mathcal{E}_r^i \gets \textsc{InitGES}(\textsc{Center}(X), \hat{G}_i[X], S_i, \delta)$\;
}
\end{algorithm}

On returning to algorithm \ref{alg:preprocessing}, we are left with initializing the graph $\hat{G}_{i+1}$. Therefore, we simply set $\hat{V}_{i+1}$ to $P$ and use again all edges $E$. Finally, we initialize $S_{\lfloor \lg n \rfloor + 2}$ to the empty set which remains unchanged throughout the entire course of the algorithm.

Let us briefly sketch the analysis of the algorithm which is more carefully analyzed in subsequent sections. Using again $\delta = 64 \lg^2 n$ and lemma \ref{lma:split}, we ensure that $|S_{i+1}| \leq |S_i|/2$, thus $|S_i| \leq n/2^i$ for all levels $i$. The running time of executing the $\textsc{Split}(\cdot)$ procedure $\lfloor \lg n \rfloor + 1$ times incurs running time $O(m \log^4 n)$ and initializing the GES-trees takes at most $O(m \delta)$ time on each level therefore incurring running time $O(m \log^3 n)$. 

\section{Finding Separators} 
\label{subsec:separators}

Before we describe how to update the data structure after an edge deletion, we want to explain how to find good separators since it is crucial for our update procedure. We then show how to obtain an efficient implementation of the procedure $\textsc{Split}(\cdot)$ that is the core procedure in the initialization.

Indeed, the separator properties that we want to show are essentially reflected in the properties of lemma \ref{lma:split}. For simplicity, we describe the separator procedures on simple graphs instead of our graphs $\hat{G}_i$; it is easy to translate these procedures to our multi-graphs
$\hat{G}_i$ because the separator procedures are not dynamic; they are only ever invoked on a fixed graph, and so we do not have to worry about node splitting and the like.

To gain some intuition for the technical statement of our separator properties stated in lemma \ref{lma:sep}, consider that we are given a graph $G=(V,E)$, a subset $S$ of the vertices $V$, a vertex $r \in V$ and a depth $d$. Our goal is to find a separator $S_{Sep} \subseteq S$, such that every vertex in the graph $G \setminus S_{Sep}$ is either at $S$-distance at most $d$ from $r$ \emph{or} cannot be reached from $r$, i.e. is separated from $r$. 

We let henceforth $V_{Sep} \subseteq V$ denote the set of vertices that are still reachable from $r$ in $G \setminus S_{Sep}$ (in particular there is no vertex $S_{Sep}$ contained in $V_{Sep}$ and $r \in V_{Sep}$). Then, a natural side condition for separators is to require the set $S_{Sep}$ to be small in comparison to the smaller side of the cut, i.e. small in comparison to $\min\{|V_{Sep}|, |V \setminus (V_{Sep} \cup S_{Sep})|\}$.

Since we are concerned with $S$-distances, we aim for a more general guarantee: we want the set $S_{Sep}$ to be small in comparison to the number of $S$ vertices on any side of the cut, i.e. small in comparison to $\min\{|V_{Sep} \cap S|, |(V \setminus (V_{Sep} \cup S_{Sep})) \cap S|\}$. This is expressed in property \ref{prop:balanceS} of the lemma.

\begin{lemma}[Balanced Separator]
\label{lma:sep}
There exists a procedure $\textsc{OutSeparator}(r, G, S, d)$ (analogously $\textsc{InSeparator}(r, G, S, d)$) where $G=(V,E)$ is a graph, $r \in V$ a root vertex, $S \subseteq V$ and $d$ a positive integer. The procedure computes a tuple $(S_{Sep}, V_{Sep})$ such that
\begin{enumerate}
    \item $S_{Sep} \subseteq S$, $V_{Sep} \subseteq V$, $S_{Sep} \cap V_{Sep} = \emptyset$, $r \in V_{Sep}$,
    \item \label{step:separator-distance} $\forall v \in V_{Sep} \cup S_{sep}$, we have $\mathbf{dist}_G(r,v,S) \leq d$ (analogously $\mathbf{dist}_G(v,r,S) \leq d$ for $\textsc{InSeparator}(r, G, S, d)$), 
    \item \[
    |S_{Sep}| \leq \frac{\min\{|V_{Sep}\cap S|, | (V \setminus (S_{Sep} \cup V_{Sep})) \cap S|\} 2\log{n}}{d},\] \label{prop:balanceS}
    and
    \item for any $x \in V_{Sep}$ and $y \in V \setminus (S_{Sep} \cup V_{Sep})$, we have $u \not\leadsto_{G \setminus E(S_{Sep})} v$ (analogously $v \not\leadsto_{G \setminus E(S_{Sep})} u$ for $\textsc{InSeparator}(r, G, S, d)$).
\end{enumerate}
The running time of both $\textsc{OutSeparator}(\cdot)$ and $\textsc{InSeparator}(\cdot)$ can be bounded by $O(E(V_{Sep}))$.
\end{lemma}

Again, we defer the proof of the lemma \ref{lma:sep} to the appendix \ref{sec:ProofLemmaSep}, but sketch the main proof idea.

\begin{proof}
To implement procedure $\textsc{OutSeparator}(r, G, S, d)$, we start by computing a BFS at $r$. Here, we assign edges in $E_{out}(S)$ again weight $1$ and all other edges weight $0$ and say a layer consists of all vertices that are at same distance from $r$. To find the first layer, we can use the graph $G \setminus E_{out}(S)$ and run a normal BFS from $r$ and all vertices reached form the first layer $L_0$. We can then add for each edge $(u,v) \in E_{out}(S)$ with $u \in L_0$ the vertex $v$ to $L_1$ if it is not already in $L_0$. We can then contract all vertices visited so far into a single vertex $r'$ and repeat the procedure described for the initial root $r$. It is straight-forward to see that the vertices of a layer that are also in $S$ form a separator of the graph. To obtain a separator that is small in comparison to $|V_{Sep} \cap S|$, we add each of the layers $0$ to $d/2$ one after another to our set $V_{Sep}$, and output the index $i$ of the first layer that grows the set of $S$-vertices in $V_{Sep}$ by factor less than $(1+\frac{2 \log n}{d})$. We then set $S_{Sep}$ to be the vertices in $S$ that are in layer $i$. If the separator is not small in comparison to $| (V \setminus (S_{Sep} \cup V_{Sep})) \cap S|$, we grow more layers and output the first index of a layer such that the separator is small in comparison to $| (V \setminus (S_{Sep} \cup V_{Sep})) \cap S|$. This layer must exist and is also small in comparison to $|V_{Sep} \cap S|$. Because we find our separator vertices $S_{Sep}$ using a BFS from $r$, a useful property of our separator is that all the vertices in $S_{Sep}$ and $V_{Sep}$ are within bounded distance from $r$.

Finally, we can ensure that the running time of the procedure is linear in the size of the set $E(V_{Sep})$, since these are the edges that were explored by the BFS from root $r$.
\end{proof}

Let us now discuss the procedure $\textsc{Split}(G,S,d)$ that we already encountered in section \ref{subsec:Preprocessing} and whose pseudo-code is given in algorithm \ref{alg:split}. Recall that the procedure computes a tuple $(S_{Split}, P)$ such that the graph $G \setminus E(S_{Split})$ contains no SCC with $S$-diameter larger $d$ and where $P$ is the collection of all SCCs in the graph $G \setminus E(S_{Split})$.

\begin{algorithm}
\caption{$\textsc{Split}(G, S, d)$}
\label{alg:split}
\KwIn{A graph $G=(V,E)$, a set $S \subseteq V$ and a positive integer $d$.}
\KwOut{Returns a tuple $(S_{Split}, P)$, where $S_{Split} \subseteq S$ is a separator such that no two vertices in the same SCC in $G \setminus E(S)$ have $S$-distance greater than $d$. $P$ is the collection of these SCCs.}
\BlankLine

$S_{Split} \gets \emptyset; P \gets \emptyset; G' \gets G;$\;

\While{$G' \neq \emptyset$}{
    Pick an arbitrary vertex $r$ in $V$.\;
    Run in parallel $\textsc{OutSeparator}(r, G', S, d/16)$ and $\textsc{InSeparator}(r, G', S, d/16)$ and let $(S_{Sep}, V_{Sep})$ be the tuple returned by the first subprocedure that finishes.\label{lne:sepTwoWay}\;

    \lIf(\label{lne:sepTwoWayIf}){$|V_{Sep}| \leq \frac{2}{3}|V|$}{
        $(S'_{Sep}, V'_{Sep}) \gets (S_{Sep}, V_{Sep})$
    }\Else(\label{lne:sepTwoWayElse}){
        Run the separator procedure that was aborted in line \ref{lne:sepTwoWay} until it finishes and let the tuple returned by this procedure be $(S'_{Sep}, V'_{Sep})$.
    }
    
    \If(\label{line:split-if-case}){$|V'_{Sep}| \leq \frac{2}{3}|V|$} {
        $(S_{Small}, P_{Small}) \gets \textsc{Split}(G'[V'_{Sep}], V'_{Sep} \cap S, d)$\label{lne:splitRecurseIf} \;
        $S_{Split} \gets S_{Split} \cup S_{Small} \cup S'_{Sep}$\label{lne:addtoS1}\;
        $P \gets P \cup P_{Small} \cup \{ \{s\} | s \in S'_{Sep}\})$\label{lne:addToPSep}\;
        $G' \gets G'[V \setminus (V'_{Sep} \cup S'_{Sep})]$
    }\Else(\label{line:split-else-case}){ 
        \tcc{Init a generalized ES-tree from $r$ to depth $d/2$.}
        $\mathcal{E}_r^i \gets \textsc{InitGES}(r, G', S, d/2)$\;
        \tcc{Find a good separator for every vertex that is far from $r$.}
         \While(\label{line:split-else-while}) {  $(v \gets \mathcal{E}_r^i.\textsc{GetUnreachableVertex}()) \neq \bot$}  {
            \If{$\mathcal{E}_r^i.\textsc{Distance}(r,v) > d/2$}{
                $(S''_{Sep}, V''_{Sep}) \gets \textsc{InSeparator}(v,G', S, d/4)$\;
            }\Else(\tcp*[h]{If $\mathcal{E}_r^i.\textsc{Distance}(v,r) > d/2$}){
                $(S''_{Sep}, V''_{Sep}) \gets \textsc{OutSeparator}(v,G', S, d/4)$\;
            }
            $\mathcal{E}_r.\textsc{Delete}(S''_{Sep} \cup V''_{Sep})$\;
        
            $(S'''_{Sep}, P''') \gets \textsc{Split}( G[V''_{Sep}] , V''_{Sep} \cap  S , d)$\label{lne:splitRecurse}\;
            $S_{Split} \gets S_{Split} \cup S''_{Sep} \cup S'''_{Sep}$\label{lne:addtoS2}\;
            $P \gets P \cup P''' \cup \{ \{s\} | s \in S''_{Sep}\})$\label{lne:addToP1}\;
        }
        $P \gets P \cup \{\mathcal{E}_r.\textsc{GetAllVertices}()\}$\label{lne:addToP2}\;
        $G' \gets \emptyset$\;
    }
}
\Return $(S_{Split}, P)$\;
\end{algorithm}

Let us sketch the implementation of the procedure $\textsc{Split}(G,S,d)$. We first pick an arbitrary vertex and invoke the procedures $\textsc{OutSeparator}(r, G', S, d/4)$ and $\textsc{InSeparator}(r, G', S, d/4)$ to run in parallel, that is the operations of the two procedures are interleaved during the execution. If one of these subprocedures returns and presents a separator tuple $(S_{Sep}, V_{Sep})$, the other procedure is aborted and the tuple $(S_{Sep}, V_{Sep})$ is returned. If $|V_{Sep}| \leq \frac{2}{3}|V|$, then we conclude that the separator function only visited a small part of the graph. Therefore, we use the separator subsequently, but denote the tuple henceforth as $(S'_{Sep},V'_{Sep})$. Otherwise, we decide the separator is not useful for our purposes. We therefore return to the subprocedure we previously aborted and continue its execution. We then continue with the returned tuple $(S'_{Sep}, V'_{Sep})$.

From there on, there are two possible scenarios. The first scenario is that the subprocedure producing $(S'_{Sep}, V'_{Sep})$ has visited a rather small fraction of the vertices in $V$ (line \ref{line:split-if-case}); in this case, we have pruned away a small number of vertices $V'_{Sep}$ while only spending time proportional to the smaller side of the cut, so we can simply recurse on $V'_{Sep}$. We also have to continue pruning away vertices from the original set $V$, until we have either removed all vertices from $G$ by finding these separators and recursing, or until we enter the else-case (line \ref{line:split-else-case}). 

The else-case in line \ref{line:split-else-case} is the second possible scenario:  note that in this case we must have entered the else-case in line \ref{lne:sepTwoWayElse} and had both the $\textsc{InSeparator}(\cdot)$ \emph{and} $\textsc{OutSeparator}(\cdot)$ explore the large side of the cut. Thus we cannot afford to simply recurse on the smaller side of the cut $V \setminus V'_{sep}$, as we have already spent time $|V'_{sep}| > |V \setminus V'_{sep}|$. Thus, for this case we use a different approach. We observe that because we entered the else-case in line \ref{lne:sepTwoWayElse} and since we entered the else-case \ref{line:split-else-case}, we must have had that $|V_{Sep}| \geq \frac{2}{3}|V|$ \textit{and} that $|V'_{Sep}| \geq \frac{2}{3}|V|$. We will show that in this case, the root $r$ must have small $S$-distance to and at least $\frac{1}{3}|V|$ vertices. We then show that this allows us to efficiently prune away at most $\frac{2}{3}|V|$ vertices from $V$ at large $S$-distance to or from $r$. We recursively invoke $\textsc{Split}(\cdot)$ on the induced subgraphs of vertex sets that we pruned away.

We analyze the procedure in detail in multiple steps, and summarize the result in lemma \ref{lma:splitFull} that is the main result of this section. Let us first prove that if the algorithm enters the else-case in line \ref{line:split-else-case} then we add an SCC of size at least $\frac{1}{3}|V|$ to $P$. 

\begin{claim}
\label{clm:largeSCCifEStree}
If the algorithm enters line \ref{line:split-else-case} then the vertex set returned by $\mathcal{E}_r.\textsc{GetAllVertices}()$ in line \ref{lne:addToP2} is of size at least $\frac{1}{3}|V|$.
\end{claim}
\begin{proof}
Observe first that since we did no enter the if-case in line \ref{lne:splitRecurseIf}, that $|V_{Sep}| > \frac{2}{3}|V|$ and $|V'_{Sep}| > \frac{2}{3}|V|$ (since we also cannot have entered the if case in line \ref{lne:sepTwoWayIf}). 

Since we could not find a sufficiently good separator in either direction, we certified that the $S$-out-ball from $r$ defined 
\[
B_{out}(r) = \{ v \in V | \mathbf{dist}_G(r , v, S) \leq d/16\}
\]
has size greater than $\frac{2}{3}|V|$, and that similarly, the $S$-in-ball $B_{in}(r)$ of $r$ has size greater than $\frac{2}{3}|V|$. This implies that 
\[
|B_{out}(r) \cap B_{in}(r)| > \frac{1}{3}|V|.
\]
Further, we have that every vertex on a shortest-path between $r$ and a vertex $v \in B_{out}(r) \cap B_{in}(r)$ has a shortest-path from and to $r$ of length at most $d/16$. Thus the $S$-distance between any pair of vertices in $B_{out}(r) \cap B_{in}(r)$ is at most $d/8$. Now, let $SP$ be the set of all vertices that are on a shortest-path w.r.t. $S$-distance between two vertices in $B_{out}(r) \cap B_{in}(r)$. Clearly, $B_{out}(r) \cap B_{in}(r) \subseteq SP$, so $|SP| \geq |V|/3$. It is also easy to see that $G[SP]$ has $S$-diameter at most $d/4$.

At this point, the algorithm repeatedly finds a vertex $v$ that is far from $r$ and finds a separator from $v$. We will now show that the part of the cut containing $v$ is always disjoint from $SP$; since $|SP| > |V|/3$, this implies that at least $|V|/3$ vertices remain in $\mathcal{E}_r$.

Consider some vertex $v$ chosen in line \ref{line:split-else-while}. Let us say that we now run $\textsc{InSeparator}(v,G',S,d/4)$; the case where we run $\textsc{OutSeparator}(v,G',S,d/4)$ is analogous. Now, by property \ref{step:separator-distance} in lemma \ref{lma:sep}, every $s \in S_{Sep}$ has $\mathbf{dist}(s,v,S) \leq d/4$. Thus, since we only run the \textsc{InSeparator} if we have $\mathbf{dist}(r,v,S) > d/2$, we must have $\mathbf{dist}(r,s,S) > d/4$. 
\end{proof}

We point out that claim \ref{clm:largeSCCifEStree} implies that $\textsc{Split}(\cdot)$ only recurses on disjoint subgraphs containing at most a $2/3$ fraction of the vertices of the given graph. To see this, observe that we either recurse in line \ref{lne:splitRecurseIf} on $G'[V'_{Sep}]$ after we explicitly checked whether $|V'_{Sep}| \leq \frac{2}{3}|V|$ in the if-condition, or we recurse in line \ref{lne:splitRecurse} on the subgraph pruned from the set of vertices that $\mathcal{E}_r$ was initialized on. But since by claim \ref{clm:largeSCCifEStree} the remaining vertex set in $\mathcal{E}_r$ is of size at least $|V|/3$, the subgraphs pruned away can contain at most $\frac{2}{3}|V|$ vertices.

We can use this observation to establish correctness of the $\textsc{Split}(\cdot)$ procedure.

\begin{claim}
\label{clm:SplitCorrectness}
$\textsc{Split}(G, S, d)$ returns a tuple $(S_{Sep}, P)$ where $P$ is a partition of the vertex set $V$ such that 
\begin{enumerate}
    \item for $X \in P$, and vertices $u,v \in X$ we have  $\mathbf{dist}_{G \setminus E(S_{Sep})}(u,v,S) \leq d$, and
    \item for distinct $X, Y \in P$, with vertices $u \in X$ and $v \in Y$,   $u \not\rightleftarrows_{G \setminus E(S_{Sep})} v$, and
    \item 
    $|S_{Split}| \leq  \frac{32 \log n}{d} \sum_{X \in P}  \lg (n  - |X \cap S|) |X \cap S|$. \label{prop:splitcorrect3}
\end{enumerate}
\end{claim}
\begin{proof}
Let us start with the first two properties which we prove by induction on the size of $|V|$ where the base case $|V|=1$ is easily checked. For the inductive step, observe that each SCC $X$ in the final collection $P$ was added to $P$ in line \ref{lne:addToPSep}, \ref{lne:addToP1} or \ref{lne:addToP2}. We distinguish by 3 cases: 
\begin{enumerate}
    \item a vertex $s$ was added as singleton set after appearing in a separator $S_{Sep}$ but then $\{s\}$ is strongly-connected and $s$ cannot reach any other vertex in $G \setminus E(S_{Sep})$ since it has no out-going edges, or
    \item an SCC $X$ was added as part of a collection $P'''$ in line  \ref{lne:addToP1}. But then we have that the collection $P'''$ satisfies the properties in $G[V''_{Sep}]$ by the induction hypothesis and since $V''_{Sep}$ was a cut side and $S''_{Sep}$ added to $S_{out}$, we have that there cannot be a path to \emph{and} from any vertex in $G \setminus E(S_{out})$, or
    \item we added the non-trivial SCC $X$ to $P$ after constructing an GES-tree from some vertex $r \in X$ and after pruning each vertex at $S$-distance to/from $r$ larger than $d/2$ (see the while loop on line \ref{line:split-else-while}). But then each vertex that remains in $X$ can reach $r$ within $S$-distance $d/2$ and is reached from $r$ within distance $d/2$ implying that any two vertices $u,v \in X$ have a path from $u$ to $v$ of $S$-distance at most $d$.
\end{enumerate}

Finally, let us upper bound the number of vertices in $S_{Split}$. We use a classic charging argument and argue that each time we add a separator $S_{Sep}$ to $S_{Split}$ with sides $V_{Sep}$ and $V \setminus (V_{Sep} \cup S_{Sep})$ at least one of these sides contains at most half the $S$-vertices in $V \cap S$. Let $X$ be the smaller side of the cut (in term of $S$-vertices) then by property \ref{prop:balanceS} from lemma \ref{lma:sep}, we can charge each $S$ vertex in $X$ for $\frac{32 \log{n}}{d}$ separator vertices (since we invoke $\textsc{OutSeparator}(\cdot)$ and $\textsc{InSeparator}(\cdot)$ with parameter at least $d/16$). 

Observe that once we determined that a separator $S_{Sep}$ that is about to be added to $S_{Split}$ in line \ref{lne:addtoS1} or \ref{lne:addtoS2}, we only recurse on the induced subgraph $G'[V_{Sep}]$ and let the graph in the next iteration be $G'[V \setminus (V_{Sep} \cup S_{Sep})$. 

Let $X$ be an SCC in the final collection $P$. Then each vertex $v \in X$ can only have been charged at most $\lg (n  - |X \cap S|)$ times. The lemma follows.
\end{proof}

It remains to bound the running time. Before we bound the overall running time, let us prove the following claim on the running time of invoking the separator procedures in parallel.

\begin{claim}
\label{clm:twowaysep}
We spend $O(E(V'_{Sep} \cup S'_{Sep}))$ time to find a separator in line \ref{lne:sepTwoWayIf} or \ref{lne:sepTwoWayElse}.
\end{claim}
\begin{proof}
Observe that we run the subprocedures $\textsc{OutSeparator}(r, G', S, d/16)$ and $\textsc{InSeparator}(r, G', S, d/16)$ in line \ref{lne:sepTwoWay} in parallel. Therefore, when we run them, we interleave their machine operations, computing one operation from $\textsc{OutSeparator}(r, G', S, d/16)$ and then one operation from $\textsc{InSeparator}(r, G', S, d/16)$ in turns. Let us assume that $\textsc{OutSeparator}(r, G', S, d/16)$ is the first subprocedure that finishes and returns tuple $(S_{Sep}, V_{Sep})$. Then, by lemma \ref{lma:sep}, the subprocedure used $O(E(V_{Sep} \cup S_{Sep}))$ time. Since the subprocedure $\textsc{InSeparator}(r, G', S, d/16)$ ran at most one more operation than $\textsc{OutSeparator}(r, G', S, d/16)$, it also used $O(E(V_{Sep} \cup S_{Sep}))$ operations. A symmetric argument establishes the same bounds, if $\textsc{InSeparator}(r, G', S, d/16)$ finishes first. The overhead induced by running two procedures in parallel can be made constant. 

Since assignments take constant time, the claim is vacuously true by our discussion if the if-case in line \ref{lne:splitRecurseIf} is true. Otherwise, we compute a new separator tuple by continuing the execution of the formerly aborted separator subprocedure. But by the same argument as above, this subprocedure's running time now clearly dominates the running time of the subprocedure that finished first in line \ref{lne:sepTwoWay}. The time to compute $(S'_{Sep}, V'_{Sep})$ is thus again upper bounded by $O(E(V'_{Sep}))$ by lemma \ref{lma:sep}, as required.
\end{proof}

Finally, we have established enough claims to prove lemma \ref{lma:splitFull}. 

\begin{lemma}[Strengthening of lemma \ref{lma:split}]
\label{lma:splitFull}
$\textsc{Split}(G, S, d)$ returns a tuple $(S_{Split}, P)$ where $P$ is a partition of the vertex set $V$ such that 
\begin{enumerate}
    \item for $X \in P$, and vertices $u,v \in X$ we have  $\mathbf{dist}_{G \setminus E(S_{Split})}(u,v,S) \leq d$, and
    \item for distinct $X, Y \in P$, with vertices $u \in X$ and $v \in Y$,   $u \not\rightleftarrows_{G \setminus E(S_{Split})} v$, and
    \item 
    $|S_{Split}| \leq  \frac{32 \log n}{d} \sum_{X \in P}  \lg (n  - |X \cap S|) |X \cap S|$
\end{enumerate}
The algorithm runs in time $O\left(d \sum_{X \in P} (1 + \lg( n  - |X|)) E(X) \right)$.
\end{lemma}
\begin{proof}
Since correctness was established in lemma \ref{clm:SplitCorrectness}, it only remains to bound the running time of the procedure. Let us first bound the running time without recursive calls to procedure $\textsc{Split}(G, S,d)$. To see that we only spend $O(|E(G)| d)$ time in $\textsc{Split}(G, S,d)$ excluding recursive calls, observe first that we can find each separator tuple $(S'_{Sep}, V'_{Sep})$ in time $O(E(V'_{Sep}))$ by claim \ref{clm:twowaysep}. We then, either recurse on $G'[V'_{Sep}])$ and remove the vertices $V'_{Sep} \cup S'_{Sep}$ with their incident edges from $G'$ or we enter the else-case (line \ref{line:split-else-case}). Clearly, if our algorithm never visits the else-case, we only spend time $O(|E(G)|)$ excluding the recursive calls since we immediately remove the edge set that we found in the separator from the graph. 

We further observe that the running time for the GES-tree can be bounded by $O(|E(G)| d)$. The time to compute the separators to prune vertices away from the GES-tree is again combined at most $O(|E(G)|)$ by lemma \ref{lma:sep} and the observation that we remove edges from the graph $G$ after they were scanned by one such separator procedure.

We already discussed that claim \ref{clm:largeSCCifEStree} implies that we only recurse on disjoint subgraphs with at most $\frac{2}{3}|V|$ vertices. We obtain that each vertex in a final SCC $X$ in $P$ participated in at most $O(\log( n - |X|))$ levels of recursion and so did its incident edges hence we can then bound the total running time by $O\left(d \sum_{X \in P} (1 + \log( n  - |X|)) E(X) \right)$.
\end{proof}

\section{Handling deletions}
\label{subsec:delete}

Let us now consider how to process the deletion of an edge $(u,v)$ which we describe in pseudo code in algorithm \ref{alg:delete}. We fix our data structure in a bottom-up procedure where we first remove the edge $(u,v)$ if it is contained in any induced subgraph $\hat{G}_i[X]$ from the GES $\mathcal{E}_{\textsc{Center}(X)}$. 

\begin{algorithm}
\caption{$\textsc{Delete}(u,v)$}
\label{alg:delete}
\KwIn{An edge $(u,v) \in E$.}
\KwResult{Updates the data structure such that queries for the graph $G \setminus \{ (u,v)\}$ can be answered in constant time.}
\BlankLine

\For{ $i = 0 $ \KwTo $ \lfloor \log{n} \rfloor$}{
    \If{If there exists an $X \in \hat{V}_{i+1}$ with $u,v \in X$}{
        $\mathcal{E}_{\textsc{Center}(X)}.\textsc{Delete}(u,v)$\;
    } 
    \While{there exists an $X \in \hat{V}_{i+1}$ with $\mathcal{E}_{\textsc{Center}(X)}.\textsc{GetUnreachable}() \neq \bot$}{
        $X' \gets \mathcal{E}_{\textsc{Center}(X)}.\textsc{GetUnreachable}()$\;
        
        \tcc{Find a separator from $X'$ depending on whether $X'$ is far to reach from $r$ or the other way around.}
        \If{$\mathcal{E}_{\textsc{Center}(X)}.\textsc{Distance}(\textsc{Center}(X),X') > \delta$}{
            $(S_{Sep}, V_{Sep}) \gets \textsc{InSeparator}(X', \hat{G}_i[X], X \cap S_i, \delta/2)$ \label{lne:DelSepIn}
        }
        \Else(\tcp*[h]{$\mathcal{E}_{\textsc{Center}(X)}.\textsc{Distance}(X', \textsc{Center}(X)) > \delta$}){
            $(S_{Sep} , V_{Sep}) \gets \textsc{OutSeparator}(X' , \hat{G}_i[X] , X \cap S_i , \delta/2)$\label{lne:DelSepOut}
        }

        \tcc{If the separator is chosen such that $V_{Sep}$ is small, we have a good separator, therefore we remove $V_{Sep}$ from $\mathcal{E}_r$ and maintain the SCCs in $\hat{G}_{i}[V_{Sep}]$ separately. Otherwise, we delete the entire GES $\mathcal{E}_{\textsc{Center}(X)}$ and partition the graph with a good separator.}
        \If{$|\textsc{Flatten}(V_{Sep})| \leq \frac{2}{3}|\textsc{Flatten}(X)|$}{
            $\mathcal{E}_{\textsc{Center}(X)}.\textsc{Delete}(V_{Sep}\cup S_{Sep})$\;
            $(S'_{Sep}, P') \gets \textsc{Split}(\hat{G}[V_{Sep}], V_{Sep} \cap S_i, \delta/2)$\label{lne:DelSplit1}\;
            $S''_{Sep} \gets S_{Sep} \cup S'_{Sep}$\;
            $P'' \gets P' \cup  S_{Sep}$\;
        }
        \Else{
            $\mathcal{E}_{\textsc{Center}(X)}.\textsc{Delete}()$\label{lne:ESdelete}\;
            $(S''_{Sep}, P'') \gets \textsc{Split}(\hat{G}_i[X], X \cap S_i, \delta/2)$\label{lne:DelSplit2}\;
        }

        \tcc{After finding the new partitions, we init them, execute the vertex splits on the next level and add the separator vertices.}
        $\textsc{InitNewPartition}(P'', i, \delta)$\;
        
        \ForEach{$Y \in P''$}{
            $\mathcal{E}_{\textsc{Center}(X)}.\textsc{SplitNode}(Y)$\;
        }
        $\mathcal{E}_{\textsc{Center}(X)}.\textsc{Augment}(S''_{Sep})$\label{lne:augmentInDelete}\;
        $S_{i+1} \gets S_{i+1} \cup S''_{Sep}$\;
    }
}
\end{algorithm}

Then, we check if any GES $\mathcal{E}_{\textsc{Center}(X)}$ on a subgraph $\hat{G}_i[X]$ contains a node that became unreachable due to the edge deletion or the fixing procedure on a level below. Whilst there is such a GES $\mathcal{E}_{\textsc{Center}(X)}$, we first find a separator $S_{Sep}$ from $X'$ in lines \ref{lne:DelSepIn} or \ref{lne:DelSepOut}. We now consider two cases based on the size of the set $\textsc{Flatten}(V_{Sep})$. Whilst focusing on the size of $\textsc{Flatten}(V_{Sep})$ instead of the size of $V_{Sep}$ seems like a minor detail, it is essential to consider the underlying vertex set instead of the node set, since the node set can be further split by node split updates from lower levels.

Now, let us consider the first case, when the set $V_{Sep}$ separated by $S_{Sep}$ is small (with regard to $\textsc{Flatten}(V_{Sep})$); in this case, we simply prune $V_{Sep}$ from our tree by adding $S_{Sep}$ to $S_{i+1}$, and then invoke $\textsc{Split}(\hat{G}_i[V_{Sep}], V_{Sep} \cap S_i, \delta/2)$ to get a collection of subgraphs $P'$ where each subgraph $Y \in P'$ has every pair of nodes $A, B \in Y$ at $S_{i}$-distance $\delta/2$. (We can afford to invoke $\textsc{Split}$ on the vertex set $V_{Sep}$ because we can afford to recurse to on the smaller side of a cut.)

The second case is when $V_{Sep}$ is large compared to the number of vertices in node set of the GES-tree. In this case we do not add $S_{Sep}$ to $S_{i+1}$. Instead we we declare the GES-tree $\mathcal{E}_{\textsc{Center}(X)}$ invalid, and delete the entire tree. We then partition the set $X$ that we are working with by invoking the $\textsc{Split}$ procedure on all of $X$. (Intuitively, this step is expensive, but we will show that whenever it occurs, there is a constant probability that the graph has decomposed into smaller SCCs, and we have thus made progress.)


Finally, we use the new partition and construct on each induced subgraph a new GES-tree at a randomly chosen center. This is done by the procedure $\textsc{InitNewPartition}(P', i, \delta)$ that was presented in subsection \ref{subsec:Preprocessing}. We then apply the updates to the graph $\hat{G}_{i+1}$ using the GES-tree operations defined in lemma \ref{lma:AugmentedGES}. Note, that we include the separator vertices as singleton sets in the partition and therefore invoke $\mathcal{E}_X.\textsc{SplitNode}(\cdot)$ on each singleton before invoking $\mathcal{E}_X.\textsc{Augment}(S''_{Sep})$ which ensures that the assumption from lemma \ref{lma:AugmentedGES} is satisfied. As in the last section, let us prove the following two lemmas whose proofs will further justify some of the details of the algorithm. 

We start by showing that because we root the GES-tree for SCC $X$ at a \emph{random} root $r$, if the GES-tree ends up being deleted in \ref{lne:ESdelete} in algorithm \ref{alg:delete}, this means that with constant probability $X$ has decomposed into smaller SCCs, and so progress has been made.

\begin{lemma}[c.f. also \cite{chechik2016decremental}, Lemma 13]
\label{lma:EStreeprob}
Consider an GES $\mathcal{E}_r = \mathcal{E}_{\textsc{Center}(X)}$ that was initialized on the induced graph of some node set $X_{Init}$, with $X \subseteq X_{Init}$, and that is deleted in line \ref{lne:ESdelete} in algorithm \ref{alg:delete}. Then with probability at least $\frac{2}{3}$, the partition $P''$ computed in line \ref{lne:DelSplit2} satisfies that each $X' \in P''$ has $|\textsc{Flatten}(X')| \leq \frac{2}{3}|\textsc{Flatten}(X_{Init})|$.
\end{lemma}
\begin{proof}
Let $i$ be the level of our hierarchy on which $\mathcal{E}_{r}$ was initialized, i.e. $\mathcal{E}_{r}$ was initialized on graph $\hat{G}_i[X_{Init}]$, and went up to depth $\delta$ with respect to $S_i$-distances (see Algorithm \ref{alg:newPart}). 

Let $u_1, u_2, ..$ be the sequence of updates since the GES-tree $\mathcal{E}_r$ was initialized that were either adversarial edge deletions, nodes added to $S_i$ or node splits in the graph $\hat{G}_i[X_{Init}]$. Observe that this sequence is independent of how we choose our random root $r$, since they occur at a lower level, and so do not take any GES-trees at level $i$ into account. Recall, also, that the adversary cannot learn anything about $r$ from our answers to queries because the SCCs of the graph are objective, and so do not reveal any information about our algorithm. We refer to the remaining updates on $\hat{G}_i[X_{Init}]$ as \textit{separator} updates, which are the updates adding nodes to $S_{i+1}$ and removing edges incident to $S_{i+1}$ or between nodes that due to such edge deletions are no longer strongly-connected. We point out that the separator updates are heavily dependent on how we chose our random source. The update sequence that the GES-tree undergoes up to its deletion in line \ref{lne:ESdelete} is a mixture of the former updates that are independent of our chosen root $r$ and the separator updates.

Let $G^j$ be the graph $\hat{G}_i$ after the update sequence $u_1, u_2, ..., u_j$ is applied. Let $X_{max}^j$ be the component of $S_i$-diameter at most $\delta/2$ that maximizes the cardinality of $\textsc{Flatten}(X_{max}^j)$ in $G^j$. We choose $X_{max}^j$ in this way because we want to establish an upper bound on the largest SCC of $S_i$-diameter at most $\delta/2$ in $G^j$. We then show that that if a randomly chosen source deletes a GES-tree (see line \ref{lne:ESdelete}) after $j$ updates, then there is a good probability that $X_{max}^j$ is small. Then by the guarantees of lemma \ref{lma:split}, the $\textsc{Split}(\cdot)$ procedure in line \ref{lne:DelSplit2} partitions the vertices into SCCs $X'$ of $S_i$-diameter at most $\delta/2$, which all have small $|\textsc{Flatten}(X')|$ because $X_{max}^j$ is small. 

More precisely, let $G^j_r$, be the graph is obtained by applying \emph{all} updates up to update $u_j$ to $\hat{G}_i[X_{Init}]$; here we include the updates $u_1, ..., u_j$, as well as all separator updates up to the time when $u_j$ takes place. (Observe that $G^j$ is independent from the choice of $r$, but $G^j_r$ is not.) Let $X_{max, r}^j$ be the component of $S_i$-diameter at most $\delta/2$ that maximizes the cardinality of $\textsc{Flatten}(X^j_{max, r})$ in this graph $G^j_r$. It is straight-forward to see that since $S_i$-distances can only increase due to separator updates, we have $|\textsc{Flatten}(X_{max, r}^j)| \leq |\textsc{Flatten}(X_{max}^j)|$ for any $r$. Further $|\textsc{Flatten}(X_{max, r}^j)|$ upper bounds the size of any component $X' \in P''$, i.e. $|\textsc{Flatten}(X')| \leq |\textsc{Flatten}(X_{max, r}^j)|$ if the tree $\mathcal{E}_r$ is deleted in line \ref{lne:ESdelete} while  handling update $u_j$; the same bound holds if $\mathcal{E}_r$ is deleted after update $u_j$, because the cardinality of $\textsc{Flatten}(X_{max, r}^j)$ monotonically decreases in $j$, i.e. $|\textsc{Flatten}(X_{max, r}^{j})| \leq |\textsc{Flatten}(X_{max, r}^{j-1})|$ since updates can only increase $S_i$-distances. 

Now, let $k$ be the index, such that 
\[
|\textsc{Flatten}(X_{max}^k)| \leq \frac{2}{3}|\textsc{Flatten}(X_{Init})| < |\textsc{Flatten}(X_{max}^{k-1})|.
\]
i.e. $k$ is chosen such that after the update sequence $u_1, u_2,..., u_{k}$ were applied to $\hat{G}_i[X_{Init}]$, there exists no SCC $X$ in $G^k$ of diameter at most $\delta/2$ with $|\textsc{Flatten}(X)| > \frac{2}{3}|\textsc{Flatten}(X_{Init})|$. 

In the remainder of the proof, we establish the following claim: if we chose some vertex $r \in \textsc{Flatten}(X^{k-1}_{max})$, then the GES-tree would not be been deleted before update $u_k$ took place. Before we prove this claim, let us point out that this implies the lemma: observe that by the independence of how we choose $r$ and the update sequence $u_1, u_2, ..$, we have that 
\[
Pr[r \in X_{max}^{k-1} | u_1, u_2, ..] = Pr[r \in X_{max}^{k-1}] = \frac{|\textsc{Flatten}(X_{max}^{k-1})|}{|\textsc{Flatten}(X_{Init})} > \frac{2}{3}
\]
where the before-last equality follows from the fact that we choose the root uniformly at random among the vertices in $\textsc{Flatten}(X_{Init})$. Thus, with probability at least $\frac{2}{3}$, we chose a root whose GES-tree is deleted during or after the update $u_k$ and therefore the invoked procedure $\textsc{Split}(\cdot)$ ensures that every SCC $X' \in P''$ satisfies $|\textsc{Flatten}(X')| \leq  |\textsc{Flatten}(X_{max}^k)| \leq \frac{2}{3}|\textsc{Flatten}(X_{Init})|$, as required. 

Now, let us prove the final claim. We want to show that if $r \in X^{k-1}_{max}$, then the GES-tree would not have been deleted before update $u_k$. To do so, we need to show that even if we include the separator updates, the SCC containing $r$ continues to have size at least $\frac{2}{3}|\textsc{Flatten}(X_{Init})|$ before update $u_k$. In particular, we argue that before update $u_k$, none of the separator updates decrease the size of $X^{k-1}_{max}$. The reason is that the InSeparator computed in Line
\ref{lne:DelSepIn} of Algorithm \ref{alg:delete} is always run from a node $X$ whose $S_i$-distance from $r$ is at least $\delta$. (The argument for an OutSeparator in Line \ref{lne:DelSepOut} is analogous.)
Now, the InSeparator from $X$ is computed up to $S_i$-distance $\delta/2$, so by Property \ref{step:separator-distance} of Lemma \ref{lma:sep}, we have that all nodes pruned away from the component have $S_i$-distance at most $\delta/2$ to $X$; this implies that these nodes have $S_i$-distance more than $\delta/2$ from $r$, and so cannot be in $X^{k-1}_{max}$, because $X^{k-1}_{max}$ was defined to have $S_i$-diameter at most $\delta/2$. Thus none of the separator updates affect $X^{k-1}_{max}$ before update $u_k$, which concludes the proof of the lemma.


\end{proof}

Next, let us analyze the size of the sets $S_i$. We analyze $S_i$ using the inequality below in order to ease the proof of the lemma. We point out that the term $\lg(n -|X \cup S_i|)$ approaches $\lg n$ as the SCC $X$ splits further into smaller pieces. Our lemma can therefore be stated more easily, see therefore corollary \ref{cor:SisSmall}. 

\begin{lemma}
\label{lma:setS}
During the entire course of deletions our algorithm maintains 
\begin{align}
|S_0| &= n                    & \\
|S_{i+1}| &\leq  \frac{32 \log n}{\delta} \sum_{X \in \hat{V}_{i}}  \lg (n  - |X \cap S_{i}|) |X \cap S_{i}| &  \text{for }i \geq 0
\end{align}
\end{lemma}
\begin{proof}
We prove by induction on $i$. It is easy to see that $S_0$ has cardinality $n$ since we initialize it to the node set in procedure \ref{alg:preprocessing}, and since each set $S_i$ is an increasing set over time.

Let us therefore focus on $i > 0$. Let us first assume that the separator nodes were added by the procedure $\textsc{OutSeparator}(\cdot)$ (analogously $\textsc{InSeparator}(\cdot)$). Since the procedure is invoked on an induced subgraph $\hat{G}_i[X]$ that was formerly strongly-connected, we have that either $V_{Sep}$ or $X \setminus (V_{Sep} \cup S_{Sep})$ (or both) contain at most half the $S_i$-nodes originally in $X$. Let $Y$ be such a side. Since adding $S_{Sep}$ to $S_i$ separates the two sides, we have that RHS of the equation is increased by at least $\frac{32 \log n}{\delta} |Y \cap S_i|$ since $\lg( n - |Y \cap S_i|) |Y \cap S_i| - \lg( n - |X \cap S_i|) |Y \cap S_i| \geq |Y \cap S_i|$. Since we increase the LHS by at most $\frac{4 \log n}{\delta} |Y \cap S_i|$ by the guarantees in lemma \ref{lma:sep}, the inequality is still holds.

Otherwise, separator nodes were added due to procedure $\textsc{Split}(\cdot)$. But then we can straight-forwardly apply lemma \ref{lma:splitFull} which immediately implies that the inequality still holds.

Finally, the hierarchy might augment the set $S_i$ in line \ref{lne:augmentInDelete}, but we observe that $f(s) = \lg(n - s) * s$ is a function increasing in $s$ for $s \leq \frac{1}{2} n$ which can be proven by finding the derivative. Thus adding nodes to the set $S_i$ can only increase the RHS whilst the LHS remains unchanged.
\end{proof}

\begin{corollary}
\label{cor:SisSmall}
During the entire course of deletions, we have $|S_{i+1}| \leq \frac{16 \lg^2 n}{\delta} |S_{i}|$, for any $i \geq 0$. 
\end{corollary}

\section{Putting it all together}
\label{sec:alltogether}

By corollary \ref{cor:SisSmall}, using $\delta = 64 \lg^2 n$, we enforce that each $|S_i| \leq n/2^i$, so $\hat{G}_{\lfloor \lg{n} \rfloor + 1}$ is indeed the condensation of $G$. Thus, we can return on queries asking whether $u$ and $v$ are in the same SCC of $G$, simply by checking whether they are represented by the same node in $\hat{G}_{\lfloor \lg{n} \rfloor + 1}$ which can be done in constant time. 

We now upper bound the running time of our algorithm by $O(m \log^5 n)$ and then refine the analysis slightly to obtain the claimed running time of $O(m \log^4 n)$.

By lemma \ref{lma:EStreeprob}, we have that with probability $\frac{2}{3}$, that every time a node leaves a GES, its induced subgraph contains at most a fraction of $\frac{2}{3}$ of the underlying vertices of the initial graph. Thus, in expectation each vertex in $V$ participates on each level in $O(\log n)$ GES-trees. Each time it contributes to the GES-trees running time by its degree times the depth of the GES-tree which we fixed to be $\delta$. Thus we have expected time $O(\sum_{v \in V} \mathbf{deg}(v) \delta \log n) = O(m \log^3 n)$ to maintain all the GES-trees on a single level by lemma \ref{lma:AugmentedGES}. There are $O(\log n)$ levels in the hierarchy, so the total expected running time is bounded by $O(m \log^4 n)$.

By lemmas \ref{lma:splitFull}, the running time for invoking $\textsc{Split}(G[X], S, \delta/2)$ can be bounded by $O(E(X) \delta \log n) = O(E(X) \log^3 n)$. After we invoke  $\textsc{Split}(G[X], S, \delta/2)$ in algorithm \ref{alg:delete}, we expect with constant probability again by lemma \ref{lma:EStreeprob}, that each vertex is at most $O(\log n)$ times in an SCC on which the $\textsc{Split}(\cdot)$ procedure is invoked upon. We therefore conclude that total expected running time per level is $O(m \log^4 n)$, and the overall total is $O(m \log^5 n)$. 

Finally, we can bound the total running time incurred by all invocations of $\textsc{InSeparator}(\cdot)$ and $\textsc{OutSeparator}(\cdot)$ outside of $\textsc{Split}(\cdot)$ by the same argument and obtain total running time $O(m \log^2 n)$ since each invocation takes time $O(E(G))$ on a graph $G$.

This completes the running time analysis, establishing the total expected running time $O(m \log^5 n)$. We point out that the bottleneck of our algorithm are the invocations of $\textsc{Split}(\cdot)$. We can reduce the total expected cost of these invocations to $O(m \log^4 n)$ by using the more involved upper bound on the running time of $\textsc{Split}(\cdot)$ of 
\[
O\left(\delta E(X) + \delta \sum_{X \in P} \log( n  - |\textsc{Flatten}(X)|) E(X) \right)
\]
where we use $|\textsc{Flatten}(X)|$ instead of $|X|$ to capture node splits. Then, we can bound the costs incurred by the first part of the bound by $O(m \log^4 n)$; for the second part we also get a bound $O(m \log^4 n)$ by using a telescoping sum argument on the size of the graph.

This concludes our proof of theorem \ref{thm:SCCmain}.

\section{Conclusion}

In this article, we presented the first algorithm that maintains SCCs or SSR in decremental graphs in almost-linear expected total update time $\tilde{O}(m)$. Previously, the fastest algorithm for maintaining the SCCs or SSR achieved expected total update time $\tilde{O}(m \sqrt{n})$. Three main open questions arise in the context of our new algorithm:
\begin{itemize}
    \item Can the complexity of the Single-Source Shortest-Path problem in decremental directed graphs be improved beyond total expected update time of $O(mn^{0.9 + o(1)})$\cite{henzinger2014sublinear, henzinger2015improved} and can it even match the time complexity achieved by our algorithm? 
    \item Does there exist a \textit{deterministic} algorithm to maintain SCCs/SSR in a decremental graph beyond the $O(mn)$ total update time complexity?
    \item And finally, is there a algorithm that solves All-Pairs Reachability in fully-dynamic graphs with update time $\tilde{O}(m)$ and constant query time? Such an algorithm is already known for dense graphs\cite{demetrescu2004new} but is still unknown for graphs of sparsity $O(n^{2-\epsilon})$.  
\end{itemize}

\paragraph{Acknowledgements}

The second author of the paper would like to thank Jacob Evald for some helpful comments on organization and correctness.

\pagebreak

\printbibliography[heading=bibintoc] 

\pagebreak

\appendix

\section{Fully-Dynamic Reachability from many sources}
\label{sec:fullyReach}

\begin{theorem}
\label{thm:APReach}
Given a directed graph $G=(V,E)$, $S \subseteq V$, we can maintain a data structure that supports the operations:
\begin{itemize}
    \item $\textsc{Insert}(u,E_u)$: Adds $u$ to $V$ with edges $E_u$ where each edge in $E_u$ has an endpoint $u$.
    \item $\textsc{Delete}(u)$:  Deletes $u$ from $V$. 
    \item $\textsc{Query}(u,v)$: For $u \in S, v \in V$, returns whether $u$ has a path to $v$.
\end{itemize}
with expected amortized update time $O\left(\frac{|S|m\log^4{n}}{t}\right)$ and worst-case query time $O(t)$ for any $t \in [1,|S|]$. The bound holds against an oblivious adaptive adversary.
\end{theorem}

The algorithm to implement theorem \ref{thm:APReach} can be described as follows: Initiate the data structure by constructing a decremental SSR structure as described in theorem \ref{thm:SCCmain} from each vertex $s \in S$. Initialize a set $I$ to be the empty set. Consider an update: if a deletion occurs, remove the edges incident to $u$ from each SSR structure, if an insertion occurs reinitialize the SSR structures at $u$ on the most recent graph $G$ and add $u$ to $I$. If $|I| \geq t$, reinitialize the entire data structure. On query for $(u,v)$, $u \in S$, check the SSR structure at $u$ whether it reaches $v$. Then check for each $r$ in $I$ whether $u$ reaches $r$ and $r$ reaches $v$. If any of these queries is successful, return that $u$ reaches $v$. Otherwise, return that $u$ cannot reach $v$. It is straight-forward to argue for correctness, and the running time is attained since the SSR structure require $O(|S|m \log^4{n})$ total time which can be amortized over $t$ updates. Each insertion additionally requires time $O(m \log^4{n})$ to (re)initialize the SSR structure at $u$. The query time is clearly $O(t)$ since we check reachability from and to each $r \in I$ in constant time. The theorem follows.

\section{Proof of Lemma \ref{lma:SimpleGES}}
\label{sec:proofsimpleges}
\begin{proof}
We define the edge weight $w(u,v)$ of an edge $(u,v)$ to be $w(u,v) = 1$ if $(u,v) \in E_{out}(S)$ and $w(u,v) = 0$ otherwise. On invocation of $\textsc{InitGES}(r, G, S, \delta)$, we find a shortest-path tree $T$ to depth $\delta$ from $r$ in $G$ by running a weighted BFS. Each edge in $T$ is called a \textit{tree edge} and each edge in $E \setminus T$ is a \textit{non-tree edge}. We assign each vertex $v \in V$, a distance label $l(v)$ that is initialized to $\mathbf{dist}_G(r,v)$. We also maintain with each vertex a set $\textsc{LevelEdges}(v)$ that is initialized to $E_{in}(v) \setminus E_{out}(v)$ that is all edges in $E_{in}(v)$ after removing self-loops.

Now, consider an update of the form $\textsc{Delete}(u,v)$. If $(u,v) \not\in T$, then simply remove $(u,v)$ from $G$ and $\textsc{LevelEdges}(v)$ and return since the shortest-path tree $T$ did not change. If $(u,v) \in T$, we remove $(u,v)$ from $T$ and add $v$ to the min-queue $Q$ with value $l(v)$. We let $Q.\textsc{Min}()$ return the minimum value of any element in $Q$ and if $Q$ is empty, we let $Q.\textsc{Min}() = \infty$.

Whilst $Q.\textsc{Min}() < \delta + 1$, we obtain the vertex $v = Q.\textsc{MinElement}()$. We then iteratively extract an edge $(y,v)$ from $\textsc{LevelEdges}(v)$ and check if $l(v) = l(y) + w(y,v)$. We stop if we found such an edge or if $\textsc{LevelEdges}(v)$ is empty. In the former case, we add $(y,v)$ to $T$ and remove $v$ from $Q$ since $T$ contains a new path of length $l(v)$ from $r$ to $v$. In the latter case, we know that $v$ cannot be reconnected in $T$ at the same distance. We therefore increase the level $l(v)$ by one. We then remove all edges $(v, z) \in T \cap E_{out}(v)$ from $T$ and add each $z$ to the queue $Q$ with its current level $l(z)$ as value. We also update the level of $v$ in $Q$ and if $l(v) \leq \delta$, we add all edges in $E_{in}(v) \setminus E_{out}(v)$ again to $\textsc{LevelEdges}(v)$. Once $Q.\textsc{Min}() \geq \delta + 1$, we remove all remaining vertices from $Q$, set their distance level to $\infty$ and remove their incident edges from the graph $G$.

We observe that the queue $Q$ can be implemented with constant operation time using two lists since we never add a vertex with level larger than the current min-elements level plus one. Clearly, we check each edge $(y,v)$ in $\textsc{LevelEdges}(v)$ once for any level of $v$ and the condition $l(v) = l(y) + w(y,v)$ can be checked in constant time. The set $\textsc{LevelEdges}(v)$ partitions the edges and since removing edges in $E_{out}(v) \cap T$ if $l(v)$ increases can be amortized over edge checks for each disconnected vertex, we conclude that the total running time can be bound by $O(\sum_{v \in V} E_{in}(v) \delta) = O(m \delta)$.

Let us now argue for correctness of the algorithm. We observe that we only delete an edge $(y,v)$ from $T$ if $l(v) > l(y) + w(y,v)$ which happens if $l(y)$ is increased or $(y,v)$ is removed. After removing $(y,v)$, we take the first edge $(u, v)$ in $\textsc{LevelEdges}(v)$ that reconnects at the same level $l(v)$, i.e. if $l(v) = l(u) + w(u,v)$. Since the levels only increase over time and edge weights are fixed, we have that once $l(v) < l(u) + w(u,v)$, we cannot use $(u,v)$ to reconnect $v$ at the same level. Thus, if $\textsc{LevelEdges}(v)$ is empty, we certified that $v$ cannot be reconnected at level $l(v)$. We also observe that after the loop to fix a vertex in $Q$ finished, that for all edges in $T$, $l(v) = l(y) + w(y,v)$, because we remove an edge $(y,v)$ from $T$ if $l(y)$ increases. Now assume that when $Q.\textsc{Min}() = \infty$, that $T$ is always a tree of the vertices in $V$ that have level $< \infty$. It follows that $T$ is a shortest-path tree by the discussion above. To see that $T$ is indeed always a tree, observe that since $l(v) = l(y) + w(y,v)$ is always enforced on edges in $T$ and edge weights are non-negative, that we have for any two vertices $y,v \in V$ on a common cycle that $l(y) = l(v)$, implying that all edge weights are $0$. But since the edge weights of $E_{out}(S)$ are set to $1$, this implies that the cycle does not contain a vertex in $S$. But this contradicts the assumption that $S$ is a feedback vertex set.

Clearly, the operation $\textsc{Distance}(r,v)$ can be implemented by returning $l(v)$ in constant time. To obtain $\textsc{Distance}(v,r)$, we also construct our data structure on the reversed edge set of $G$ from $r$. To support $\textsc{Delete}(V')$, we simply delete all incident edges from vertices in $V'$ using $\textsc{Delete}(u,v)$ and remove all vertices $V'$ from $V$. The operation $\textsc{GetUnreachableVertex}()$ can be implemented by keeping a list of all vertices with level $\infty$.
\end{proof}

\section{Proof of lemma \ref{lma:AugmentedGES}}
\label{sec:proofAugmentedGES}

\begin{proof}
Let us first describe how to implement the operations efficiently. We maintain additional pointers from each vertex $v \in V$ to the set $X \in \hat{V}$ with $v \in X$ and store with each vertex in $\hat{V}$ the number of vertices it represents. For $\textsc{SplitVertex}(X)$, we use any vertex in $x \in X$ to locate the vertex $Y \in \hat{V}$ with $X \subseteq Y$. If $X= Y$, we return. Otherwise, we compare the sizes of $X$ and $Y \setminus X$ and if $X$ has larger cardinality, we invoke $\textsc{SplitVertex}(Y \setminus X)$ and return. Let us therefore assume that $X$ is smaller than $Y \setminus X$. Then, we create a new vertex $X$ in $\hat{G}$ and move all edges in $E(X)$ from $E(Y)$. In particular, we also let the set $\textsc{LevelEdge}(X) = \textsc{LevelEdges}(Y) \cap E_{in}(X)$. We might also move the tree edge in $T$ depending on whether its endpoint maps to $X$ or $Y \setminus X$. Finally, we set $l(X)$ to $l(Y)$, remove the vertices in $X$ from $Y$, add the vertices in $\hat{V}$ that have no in-coming tree edge and do not contain $r$ to the queue $Q$ and invoke the fixing process described in the preceding lemma. 

Clearly, the operation of moving the edges from $Y$ to $X$ takes time $O(E(X))$. Since each time an edge is moved, at least one of its endpoints underlying vertex set halves, we have that each edge is moved at most $O(\log X)$ times, where $X$ is the initial endpoint of larger cardinality. The second part of the upper bound of the lemma follows. 

It is also straight-forward to see that after shifting the edges, the level of the vertices $X$ and $Y \setminus X$ is at least $l(Y)$ and that each edge in $E_{in}(Y) \setminus \textsc{LevelEdges}(Y)$ was already certified to not reconnect $X$ or $Y \setminus X$ to level $l(Y)$. Thus, the fixing process is simply part of the standard fixing process and can be included in the analysis presented for the last lemma. 

Finally, we implement the operation $\textsc{Augment}(S')$ by increasing the weight of each edge $E_{out}(S' \setminus S)$ to 1. We then remove each tree edge $(s', v) \in T \cap E_{out}(S' \setminus S)$ from $T$; add each such $v$ to $Q$ and invoke the fixing procedure. It is straight-forward to show that the costs of this procedure are subsumed by the fixing procedure.
\end{proof}

\section{Proof of Lemma \ref{lma:sep}}
\label{sec:ProofLemmaSep}

\begin{proof}
Let an out-layer $L_{out}(r, G, S, i)$ be defined to the vertices that are at distance exactly $i$ from $r$, i.e. $L_{out}(r, G, S, i) = \{ v \in V | \mathbf{dist}_G(r, v , S) = i\}$. Thus, $L_{out}(r, G, S, i) \cap L_{out}(r, G, S, j) = \emptyset$ for $i \neq j$. We can compute $L_{out}(r, G, S, 0)$ in time $O(E_{out}(L_{out}(r, G, S, 0)))$ time, by initializing a queue $Q_0$ with element $r$ and then iteratively extracting a vertex $v$ from $Q_0$ and adding for each $(v,w) \in E_{out}(v)$ the vertex $w$ to $Q_0$ if $w \not\in S$ and if $w$ was not visited by the algorithm before. Initializing the queue $Q_{i+1}$ with the vertex set $L_{out}(r, G, S, i) \cap S$, we can use the same strategy for computing $L_{out}(r, G, S, i+1)$ in time $O( E_{out}(L_{out}(r, G, S, i)) + E_{out}(L_{out}(r, G, S, i+1)))$. This corresponds to running a BFS on an edge set with weights in $\{0, 1\}$. Upon finding the first $i$ such that 
\[
|L_i \cap S| \leq \frac{ \min\{\sum_{j < i} |L_j \cap S|,  |(V \setminus \sum_{j < i}  L_j) \cap S|\} 2\log n}{d}
\]
is satisfied, and then output $S_{Sep} = L_i \cap S$ and $V_{Sep} = \cup_{j \leq i} L_i \setminus S_{Sep}$. Since we never extract the vertices in $S_{Sep}$ from the queue but only look at the out-going edges from $V_{Sep}$ we conclude that the procedure takes time $O(E_{out}(V_{Sep}))$.

To see, that there exists at least one layer $i$ satisfying this inequality, let layer $i'$ be the minimal layer such that $\sum_{j < i'} |L_j \cap S| > \frac{1}{2}|S|$. If $i' \geq d/2$, then we claim there exists at least one  layer $i \leq i'$ with 
\[
|L_i \cap S| \leq \frac{ \sum_{j < i} |L_j \cap S| 2\log n}{d}
\]
Otherwise, i.e. if $i' \geq d/2$, there we claim that there exists at least one layer $i \leq i'$ with 
\[
|L_i \cap S| \leq \frac{ \sum_{j < i} |L_j \cap S| 2\log n}{d}
\]
Let us assume the former case. Consider for the sake of contradiction that we cannot find an appropriate $i \leq d/2$, i.e. $|L_i \cap S| > \frac{ \sum_{j < i} |L_j \cap S|\log{n}}{d}$ for any $i$. Then, we have $\sum_{j < i} |L_j \cap S| > (1+\frac{2\log{n}}{d})^i$ since the set grows by factor $(1+\frac{2\log{n}}{d})$ with each iteration and it would be trivially satisfied if $L_0 \cap S = \emptyset$. But then the final set must be of size greater than $n = e^{\log{n}} \leq (1+\frac{2\log{n}}{d})^{d/2}$, thus establishing a contradiction. The case where $i' \geq d/2$ is symmetric.

For point 4. in our lemma, we observe that removing a layer implies that the layers visited can not reach the layers not visited yet. Assume for contradiction that this is not true, then we have an edge $(x,y)$ with $x \in V_{Sep}$ and $y \in V \setminus (S_{Sep} \cup V_{Sep})$ but then $y$ has to be in the same or the next layer that $x$ is contained. In the first case, $y \in V_{Sep}$ which contradicts that $y \in V \setminus (S_{Sep} \cup V_{Sep})$. In the latter case, Either $y \in V_{Sep}$ or $y$ is in the removed layer $S_{Sep}$ but both contradicts again that $y \in V \setminus (S_{Sep} \cup V_{Sep})$.
\end{proof}

\end{document}